\documentclass[a4paper,11pt]{article}

\usepackage{amsmath,amsthm,amsfonts,latexsym,bm,bbm,xspace,graphicx,float,mathtools,mathdots,physics}
\usepackage{braket,caption,subcaption,ellipsis,textcomp,hhline,pifont,combelow,booktabs}
\usepackage{wrapfig,listings}
\usepackage[small]{titlesec}
\usepackage{thm-restate}

\usepackage[linesnumbered,boxed, vlined]{algorithm2e}
\DontPrintSemicolon
\SetKwInOut{Input}{Input}
\SetKwInOut{Output}{Output}

\usepackage{thmtools}
\usepackage{thm-restate}

\usepackage[linewidth=1pt]{mdframed}

\definecolor{ForestGreen}{rgb}{0.1333,0.5451,0.1333}
\definecolor{DarkRed}{rgb}{0.65,0,0}
\definecolor{Red}{rgb}{1,0,0}
\definecolor{DarkRed}{rgb}{0.5,0.1,0.1}
\definecolor{DarkBlue}{rgb}{0.1,0.1,0.5}

\declaretheorem[numberwithin=section]{theorem}
\declaretheorem[numberlike=theorem]{lemma}

\declaretheorem[numberlike=theorem]{proposition}

\declaretheorem[numberlike=theorem]{corollary}
\declaretheorem[numberlike=theorem]{claim}

\theoremstyle{definition}
\declaretheorem[numberlike=lemma]{definition}
\usepackage[textsize=scriptsize,backgroundcolor=gray!7,bordercolor=gray,linecolor=gray]{todonotes}

\newcommand{\polylog}{\ensuremath{\poly\log}}

\renewcommand{\paragraph}[1]{\medskip\noindent{\bf #1}\xspace}

\renewcommand{\epsilon}{\ensuremath{\varepsilon}}

\let\set\relax
\DeclarePairedDelimiter{\set}{\lbrace}{\rbrace}%
\DeclarePairedDelimiter{\parens}{\lparen}{\rparen}%

\let\Pr\relax
\DeclareMathOperator*{\Pr}{\ensuremath{\mathsf{Pr}}}
\DeclareMathOperator*{\Exp}{\ensuremath{{\mathsf{E}}}}
\DeclareMathOperator*{\poly}{\operatorname{poly}}

\DeclareMathOperator{\dist}{dist}

\newcommand{\calM}{\mathcal M}

\definecolor{codegreen}{rgb}{0,0.6,0}
\definecolor{codegray}{rgb}{0.5,0.5,0.5}
\definecolor{codepurple}{rgb}{0.58,0,0.82}
\definecolor{backcolour}{rgb}{0.95,0.95,0.92}

\lstdefinestyle{mystyle}{
  backgroundcolor=\color{backcolour},
  commentstyle=\color{codegreen},
  keywordstyle=\color{magenta},
  numberstyle=\tiny\color{codegray},
  stringstyle=\color{codepurple},
  basicstyle=\ttfamily\footnotesize,
  breakatwhitespace=false,
  breaklines=true,
  captionpos=b,
  keepspaces=true,
  numbers=left,
  numbersep=5pt,
  showspaces=false,
  showstringspaces=false,
  showtabs=false,
  tabsize=2
}

\DeclareMathSymbol{\mhyphen}{\mathord}{AMSa}{"39}

\newcommand{\brak}[1]{\left(#1\right)}

\newcommand{\Expp}[1]{\mathbb{E}[{#1}]}
\newcommand{\BO}[1]{O(#1)}

\newcommand{\TO}[1]{\tilde{O}(#1)}
\newcommand{\Omc}[1][1]{\Omega\brak{#1}}

\renewcommand{\sl}{{\sqrt{\log n}}}
\newcommand{\SL}{{2^{\sqrt{\log n}}}}

\makeatletter
\newcommand{\whp}{%
w.h.p.\@ifnextchar.{\@gobble}{\xspace}%
}
\makeatother

\usepackage{comment}

\newcounter{blackbox}%
\newenvironment{blackbox}[1]{%
  \refstepcounter{blackbox}%
  \begin{mdframed}[
      frametitle={Algorithm \theblackbox. #1},
      backgroundcolor=gray!25,
      splittopskip=\topskip,
      skipabove=\topsep,
      skipbelow=\topsep,
      userdefinedwidth=\textwidth
    ]%
    }{%
  \end{mdframed}%
}

\usepackage{mathrsfs}
\usepackage[normalem]{ulem}

\usepackage{graphicx}
\usepackage{adjustbox}

\usepackage[pagebackref=true]{hyperref}

\usepackage[noabbrev,capitalise,nameinlink]{cleveref}
\crefname{theorem}{Theorem}{Theorems}
\crefname{section}{Section}{Sections}
\crefname{lemma}{Lemma}{Lemmas}
\crefname{observation}{Observation}{Observations}
\crefname{algorithm}{Algorithm}{Algorithms}
\crefname{step}{Step}{Steps}
\crefname{fact}{Fact}{Facts}
\crefname{claim}{Claim}{Claims}
\crefname{conclusion}{Conclusion}{Conclusions}
\crefname{part}{Property}{Properties}
\crefname{frame}{Frame}{Frames}
\crefname{blackbox}{Black Box}{Black Boxes}

\usepackage{amssymb}
\usepackage{tcolorbox}
\usepackage{paralist}
\usepackage[margin=2.5cm]{geometry}
\usepackage[inline]{enumitem}
\author{Keren Censor-Hillel \thanks{Department of Computer Science, Technion. \texttt{ckeren@cs.technion.ac.il}. The research is supported in part by the Israel Science Foundation (grant 529/23).} 
\and
Tomer Even \thanks{Department of Computer Science, Technion. \texttt{tomer.even@campus.technion.ac.il}.} 
\and 
Maxime Flin \thanks{Department of Computer Science, Reykjavik University, Reykjavik, Iceland. \texttt{maximef@ru.is}. The research is supported in part by the Icelandic Research Foundation (grant no.~2310015).}
\and 
Magn\'{u}s M. Halld\'{o}rsson \thanks{Department of Computer Science, Reykjavik University, Reykjavik, Iceland. \texttt{magnusmh@gmail.com}. The research is supported in part by the Icelandic Research Foundation (grant no.~217965).}
}

\newcommand{\model}[1]{{\textsf{#1}}\xspace}
\newcommand{\CC}{\model{Congested\ Clique}}
\newcommand{\local}{\model{LOCAL}}
\newcommand{\LOCAL}{\model{LOCAL}}
\newcommand{\congest}{\model{CONGEST}}

\newcommand{\alg}[1]{\textsf{#1}\xspace}
\newcommand{\ReduceMis}{\alg{ReduceMIS}}
\newcommand{\ReduceMM}{\alg{ReduceMM}}
\newcommand{\OneShotReduceMIS}{\alg{OneShotMIS}}
\newcommand{\OpRoute}{\alg{OpRoute}}
\newcommand{\OneShotReduceMISNonUnif}{\alg{OneShotMISNonUniform}}

\newcommand{\twosqlg}{2^{O(\sqrt{\log n})}}

\newcommand{\betag}{\beta(G)}

\newcommand{\llbet}{\log(\log(\betag)/\sqrt{\log n}) + 1}
\newcommand{\RC}{\BO{\llbet}}

\renewcommand{\Pr}[1]{{\mathrm{Pr}}[{#1}]}
\renewcommand{\Exp}[1]{\mathbb{E} [{#1}]}

\begin{document}
\title{When MIS and Maximal Matching are Easy in the Congested Clique}
\date{}
\setcounter{tocdepth}{3}
\maketitle

\begin{abstract}
    Two of the most fundamental distributed symmetry-breaking problems are that of finding a maximal independent set (MIS) and a maximal matching (MM) in a graph. It is a major open question whether these problems can be solved in constant rounds of the all-to-all communication model of \CC, with $O(\log\log \Delta)$ being the best upper bound known (where $\Delta$ is the maximum degree).

    We explore in this paper the boundary of the feasible, asking for \emph{which graphs} we can solve the problems in constant rounds. We find that for several graph parameters, ranging from sparse to highly dense graphs, the problems do have a constant-round solution.
    In particular, we give algorithms that run in constant rounds when:
    \begin{itemize}
        \item the average degree is at most $d(G) \le 2^{O(\sqrt{\log n})}$,
        \item the neighborhood independence number is at most $\beta(G) \le 2^{O(\sqrt{\log n})}$, or
        \item the independence number is at most $\alpha(G) \le |V(G)|/d(G)^{\mu}$, for any constant $\mu > 0$.
    \end{itemize}
    Further, we establish that these are tight bounds for the known methods, for all three parameters, suggesting that new ideas are needed for further progress.
\end{abstract}

\section{Introduction}
Symmetry breaking is fundamental to distributed computing. Two of the most fundamental symmetry-breaking problems are that of computing a maximal independent set (MIS) and a maximal matching (MM) in the distributed network. We consider these problems in the \CC model, which has a machine for each graph node and in each synchronous round, each machine can send a $O(\log n)$-bit message to every other node.

There is now a fairly good understanding of the complexity of these problems in the central model of locality, \LOCAL. There is an algorithm that runs in $\BO{ \log \Delta + (\log\log n)^{O(1)} }$ rounds \cite{Ghaffari16,ghaffari2021improved,FGGKR23} that works even in the bandwidth-constrained CONGEST model and a lower bound of $\Omega\parens*{\min\set*{\frac{\log\Delta}{\log\log \Delta}, \sqrt{\frac{\log n}{\log\log n}}}}$ \cite{KMW16}
for the \LOCAL model.
Much less is known in the \CC model.

Ghaffari \cite{G17} gave a $\TO{\sqrt{\log\Delta}}$-round \CC algorithm for MIS
which beats the \local lower bounds
when $\Delta$ is small enough. It was followed by \cite{Kon18,ghaffari2018improved}
which improved the complexity to $O(\log \log \Delta)$ by using a sparsification
technique pioneered by \cite{ahn2015correlation} for correlation clustering in streaming.
For MM, the state of the art is the $O(\log\log \Delta)$-round algorithm
of Behnezhad, Hajiaghai and Harris \cite{behnezhad2023exponentially}
which implements a sparsification mechanism comparable to the one for MIS.
Meanwhile, problems such as minimum spanning tree \cite{JN18}, $(\Delta+1)$-vertex-coloring \cite{CFGUZ19} or 2-ruling set \cite{cambus2023time} have $O(1)$-round algorithms in \CC. The 2-ruling set problem in particular is a slight relaxation of the MIS problem which has a similar complexity in the \local model \cite{BBKO22sicomp,BBKO22stoc}.

Proving lower bounds in \CC is known to be hard \cite{DKO13}.
Interestingly, as the algorithm of \cite{ghaffari2018improved}
can be implemented in $O(\log\log\Delta)$ passes of streaming,
a very recent streaming lower bound by Assadi, Konrad, Naidu and Sundaresan \cite{Assadi0NS24}
shows that such algorithms cannot break the $O(\log\log n)$ barrier in general graphs.
Contrary to MIS, there exist no such results for MM because it is trivial to compute in the streaming model.
Unconditional lower bounds for MIS and MM are only known in the much weaker
\emph{broadcast congested clique}, where nodes must send the same message to everyone in a round.
In this model, Assadi, Kol and Zhang \cite{AK022} recently showed that any MIS or MM
algorithm with $(\log n)^{O(1)}$ bandwidth needs $\Omega(\log\log n)$ rounds.
Meanwhile, the best upper bound in this model is $O(\log n)$ rounds \cite{luby86}.

Turning back to \CC, we change perspective and ask:
\begin{quote}
    \emph{\textbf{When} can MIS and MM be solved in $O(1)$ rounds of \CC?}
\end{quote}
In other words, for what kind of graphs can we solve the problems fast, and for which ones are current techniques insufficient?

We consider both sparse and dense families of graphs, pinpointing properties that allow us to quickly converge toward a solution. On one hand, we focus on graphs of low-to-medium \emph{average degree}, which subsumes other sparse graphs such as low-degree or low-arboricity graphs. On the other hand, we consider graphs with sublinear \emph{independence number}, which are graphs that can be highly dense. Another related notion is \emph{neighborhood independence number}, which is the largest number of independent nodes in the neighborhood of a given node.

\paragraph{Our Results.}
We summarize our algorithmic results succinctly as follows.
Let $n$ be the number of nodes in the graph $G$, $d(G)$ its average degree, $\alpha(G)$ the size of the largest independent set, and $\beta(G) = \max_{v\in V(G)} \alpha(G[N(v)])$ the neighborhood independence number.
We say that a problem is \emph{easy} if it can be solved in $O(1)$ rounds of \CC (on the respective graph class.
\begin{tcolorbox}
    MIS and MM are easy on graphs of \emph{average degree} at most $d(G) \le \twosqlg$.
\end{tcolorbox}
\noindent
The same holds in terms of neighborhood independence:
\begin{tcolorbox}
    MIS and MM are easy on
    graphs of \emph{neighborhood independence} at most $\beta(G) \le \twosqlg$.
\end{tcolorbox}
\noindent
Additionally, we get bounds in terms of independence number.
\begin{tcolorbox}
    MIS and MM are easy on
    graphs with \emph{independence number} at most $\alpha(G) \le n/d(G)^\mu$, for any constant $\mu > 0$.
\end{tcolorbox}

For larger ranges of these parameters, we get tradeoffs, with complexities as a function of the parameters: $O(\log (\log d(G) / \sqrt{\log n})+1))$, $O(\llbet)$, and $O(\log 1/\mu + 1)$, respectively.
Note that classes of graph where those parameters are small have been studied and are known to appear in practice; see \cref{sec:related-work} for more details.

Our results are tight for all three graph parameters. More generally, we give a family of hard graphs
for which current techniques require $\Omega(\log\log n)$ rounds. The parameters $\Delta,d(G),\beta(G),\alpha(G)$ in our hard graphs are such, that our bounds for them are tight given the known methods.
Our construction provides concrete evidence of instances for which new approaches are needed.

\subsection{Technical Overview}
\paragraph{Upper bounds.}
Our starting points for computing MIS and MM are the algorithms of
\cite{ghaffari2018improved,behnezhad2023exponentially} which take $\BO{\log\log(\Delta)}$ rounds for MIS and MM, respectively. They both use degree reduction for sparsification.

For MIS, the algorithm samples a subset of vertices that induce a graph with $\BO{n}$ edges. One designated vertex collects the edges of this graph and locally computes an MIS $S$ on it. To extend the set $S$ into an MIS, the algorithm then proceeds by computing an MIS on $G[V\setminus\Gamma(S)]$, where $\Gamma(S)$ contains all vertices in $S$ or with at least one neighbor in it (i.e., vertices that are \emph{covered} by $S$).
Computing an MIS on this residual graph $G[V\setminus\Gamma(S)]$ is an easier task as it has a maximal degree of $\TO{\sqrt{\Delta}}$, \whp. This algorithm is repeated
$\BO{\log\log(\Delta)}$ times to get a residual graph of a maximum degree $\polylog(n)$, at which point it switches to simulating a \local algorithm by letting each vertex learn its $\log\log(\Delta)$ neighborhood.

\noindent
~\\\textbf{Average degree:} As a warm-up to our main technical contribution, we observe that one can push the degree-reduction method beyond the parameter of the maximum degree, which can be replaced by the \emph{average degree}.
The key idea is to sample according to the average degree, rather than the maximum degree, and leverage a concentration bound that allows for dependencies.
With efficient routing, this leads us to a constant-round algorithm when the average degree is at most $\twosqlg$, which we use next.

~\\\noindent\textbf{Independence number:}  One of our main technical contributions is to bound the running time of the repeated degree reduction algorithm in terms of the \emph{independence number} of the graph.
To give a flavor of this, consider a graph $G$ with $n$ vertices, average degree $d(G)$, and an independence number as small as possible, i.e., $\alpha(G)=n/(d(G)+1)$.
After one round of degree reduction, we get an independent set $S$ and a residual graph $H=G[V\setminus\Gamma(S)]$, where the maximal degree in $H$ is $\TO{\sqrt{d(G)}}$.
Note that since $H$ is an induced subgraph of $G$, its independence number cannot be larger than that of $G$. We use Tur\'an's bound to obtain the following bound:
$    \frac{n(H)}{2d(H)}\leq \alpha(H)\leq \alpha(G)\leq n/(d(G)+1)$.
Further, we can bound the number of edges of $H$ by
multiplying both sides by $(d(H))^2$ to get that
$   |E(H)|= n(H)\cdot d(H)/2\leq n\cdot \frac{(d(H))^2}{d(G)+1}\leq \TO{n}$,
where the last inequality follows as $d(H)=\TO{\sqrt{d(G)}}$.
This, together with efficient routing, leads us to a constant-round algorithm, as the residual graph has average degree is at most $\twosqlg$.

This toy example was only for $\alpha(G)=n/(d+1)$. We extend it by showing a $\BO{\log(1/\mu) +1}$ round algorithm for MIS, where $\mu>0$ is such that the independence number of $G$ is $\alpha(G)=n/(d+1)^\mu$.

This framework also applies for MM. Consider the algorithm of \cite{behnezhad2023exponentially}. It partitions the vertices of the graph into random sets, such that each set induces a graph with at most $\BO{n}$ edges. Then, a maximal matching is computed on each part. Let $M'$ denote the union of those computed matchings. The algorithm then proceeds with a clean-up phase, in which vertices of high degree are matched. Let $M$ denote the union of  $M'$ and the edges computed in the clean-up phase. The algorithm has the property that $G[V\setminus V(M)]$ has a maximal degree of $\TO{\Delta^{0.92}}$.

As for the MIS case, we turn this into an algorithm which reduces the maximal degree in terms of the average degree and show that if the input graph has a small independence number then we can compute an MM faster.

~\\\noindent\textbf{Neighborhood independence:}
Our second main technical contribution involves graphs with small neighborhood independence. A graph $G$ has neighborhood independence $\betag$ if the size of any independent set in the neighborhood of any vertex is at most $\betag$.
Note that $\betag$ and $\alpha(G)$ are not comparable as $\betag\leq \Delta(G)$, and $n/(\Delta(G)+1)\leq \alpha(G)\leq n\betag/\delta(G)$.

Suppose a graph has a minimum degree of at least $(d(G))^{1/2}$, and that $\betag$ is, say, at most $(d(G))^{1/4}$. Then we get that $\alpha(G)\leq n/(d(G))^{1/4}$, allowing us to use our framework for low independence number. This motivates us to find subgraphs for which this holds. However, a direct degree reduction is insufficient.
The challenge here is that the minimum degree might be much smaller than the average degree. 

Our approach for this case is to partition the vertices into $\log n $ degree classes where the $i$-th degree class,
denoted by $V_i$, contains vertices of roughly the same degrees, and then use
the fact that each $G[V_i]$ has relatively small independence number. Our final step here is to replace sampling of vertices with a uniform probability of $1/\sqrt{d(G)}$, by a \emph{non-uniform} probability of $1/\sqrt{\deg_G(v)}$ (such sampling was applied for 2-ruling sets in \cite{cambus2023time}).
We show that this algorithm has the property that within the residual graph, each node
has at most $\TO{\sqrt{x}}$ neighbors that had degrees in $G$ in the range $[x,2x]$.
This allows us to show that after one iteration, we get a residual graph $H$, where $H[V_i]$ contains only $\TO{n\betag}$ edges w.h.p, and therefore so does $H$. We can therefore find an MIS on $H$ in time $\RC$ \whp.

For matching, our algorithm in terms of $\betag$ is simpler. We partition the vertices into degree classes as before. In parallel, each class computes a matching on the edges inside the class, using the degree reduction algorithm. We claim that finding a maximal matching in each degree class takes only a constant number of rounds, after which, the residual graph $H$ contains only
$\TO{n\betag}$ edges \whp. The crucial point here is as follows. If $W$ is
a set of unmatched vertices inside the same degree class, i.e., of degree between $a$ and $2a$, then $W$ is an independent set, from which we obtain that
$\abs{W}\leq n\betag/a$.
Therefore, these vertices contribute at most $n\betag$ edges to $H$.

\paragraph{Tight Analysis for Existing Algorithms.}
We provide a hard distribution over graphs, on which the state-of-the-art methods require $\Omega(\log\log n)$ iterations to complete. Notably, the maximum degree in these graphs, the average degree, and the global- and neighborhood- independence numbers are all values for which our bounds with respect to them become tight given the known methods.\footnote{We do not know if this is the case for the construction of \cite{Assadi0NS24}.}
~\\\noindent\textbf{MIS:}
We provide a tight analysis for the running time of all algorithms that compute an MIS using the following framework:
\begin{enumerate*}[label=(\roman*)]
    \item samples a set of vertices $S$, such that the induced subgraph $G[S]$ has $\BO{n}$ edges w.h.p,
    \item computes a greedy MIS over $G[S]$,
    \item removes all covered vertices, and repeats.
\end{enumerate*}
The high-level intuition behind our construction is as follows. We split the $n$ vertices into roughly $\log\log n$ levels, and denote by $k$ the number of vertices in each level. Within level $i$, we partition the vertices into $k^{1-1/4^i}$ equal-size clusters, each forming a clique.
Note that if the algorithm sampled at least one vertex $v$ from some cluster $C$, then a vertex from $C$ joins the MIS, and all other vertices in $C$ are covered.
We say that a cluster is sampled if at least one of its vertices is sampled.
Yet, if for a given cluster, no vertex gets sampled in a certain iteration, then none joins the independent set, then none of its vertices get covered.
As the level number increases, so does its number of clusters, while their size decreases, making them less likely to be sampled.

The running time of an algorithm that samples vertices with some uniform probability might be $\Omc[\log\log n]$, but an algorithm that samples each vertex with probability $1/\sqrt{d_v}$, is likely to hit all clusters very quickly (plus it can be easily handled in a single iteration of a Luby approach).

To overcome this, we add random edges to the graph, obtaining a distribution over graphs. A random edge is added between two vertices in levels $i,j$ for $i<j$ with probability that is roughly the inverse of the number of clusters in level $i$ (up to $\polylog n$ factors).
The obtained graph becomes nearly regular, and then it is unlikely that the smaller clusters (of the higher levels) are sampled, even if the sampling probability is based on the vertices degree.
The random edges complicate the analysis as a vertex is covered when its cluster is sampled, or when a neighbor of the vertex in a different cluster is sampled.
A crucial step in this analysis requires proving that \whp at least $1-o(1)$ fraction of the vertices in the $r$-th level are not covered before iteration $r$ is over.
Since the graph has roughly $\log\log n$ levels, we conclude that at least $\Omc[\log\log n]$ iterations are needed to cover all vertices.
The full details of this proof appear in \Cref{sec:hard}.

~\\\noindent\textbf{MM:} We also provide a tight analysis for the running time of all algorithms that compute an MM using the following framework:
\begin{enumerate*}[label=(\roman*)]
    \item randomly partitions the vertices into sets $(S_1,\ldots,S_k)$,
    such that for every $i$ the induced subgraph $G[S_i]$ has $\BO{n}$ edges w.h.p,
    \item computes a greedy MM over each $G[S_i]$,
    \item invokes a clean-up phase that matches high-degree vertices,
    \item removes all covered vertices, and repeats.
\end{enumerate*}
Our hard graph distribution for MM is only a slight modification of the one for MIS: The number of vertices in each level is increased by $4$ as the level increases. This ensures that any maximal matching must ``reveal'' the last level of the graph.
We also ensure that graphs taken from this distribution have a perfect matching, by setting the number of vertices in each cluster to be even.

A crucial step in the analysis of the running time of such algorithm requires proving that \whp at least $\frac{2}{3}\cdot (1-o(1))$ fraction of the vertices in the $r$-th level are not matched before iteration $r$ is over.
Using this property and the fact that $3/4$ fraction of the vertices are in the last level, and $3/4$ fraction of the remaining vertices are in the second-to-last level, we conclude that after $\ell/2$ iterations, where $\ell$ is the number of levels, at least $\frac{8 (1-o(1))}{9}$ fraction of the vertices are not matched, and therefore the computed matching is not maximal, as any maximal matching must in a graph with a perfect matching must match at least half of the vertices.
Therefore, the running time of the algorithm must be at least $\Omc[\log\log n]$ \whp.

\subsection{Related Work}
\label{sec:related-work}

\paragraph{Algorithms in \local and \congest.}
Maximal Independent Set and Maximal Matching have long been studied in the classical \local and \congest models (see, e.g., \cite{luby86,panconesi95,kuhn2010local,barenboim2016locality,Ghaffari16,rozhovn2019polylogarithmic,fischer2020improved,balliu2021lower,ghaffari2021improved,FGGKR23,ghaffari2021improved,balliu2023distributed}).
Combining recent advances on local rounding \cite{fischer2020improved,FGGKR23} and network decomposition \cite{rozhovn2019polylogarithmic,ghaffari2021improved,GHIR23},
authors of \cite{GG23} gave a $\TO{\log^2 n}$-round
state-of-the-art deterministic \local upper bound.
Using the shattering technique \cite{barenboim2016locality,Ghaffari16} this yields a
$O(\log\Delta) + \TO{\log^2\log n}$ round \local algorithm.
As for \local lower bounds, authors of \cite{balliu2021lower} showed that any \local algorithm that solves MIS or MM with probability at least $1-1/n$ needs $\Omega\parens*{\min\set*{\Delta, \frac{\log\log n}{\log\log\log n}}}$. Hence, local algorithms are nearly optimal.

\paragraph{Algorithms for Maximum Matching in \CC.}
For matching, using different techniques,
authors of \cite{ghaffari2018improved,AssadiBBMS19}
gave an algorithm to $(1+\epsilon)$-approximate the integral maximum matching
running in $O(\log\log n)$ rounds.
Authors of \cite{behnezhad2023exponentially} showed that a simple degree reduction
algorithm could be implemented for matching, resulting in a $O(\log\log\Delta)$-round
algorithm for maximal matching.
A recent breakthrough of \cite{FMU22} shows that given a \CC algorithm
for maximal matching, one can boost the approximation factor
to $1+\epsilon$ by paying only a $(1/\epsilon)^{O(1)}$ factor
of overhead in the round complexity.

\paragraph{Algorithms for 2-ruling set.}
It is interesting to note that a weaker variant of the MIS
called \emph{2-ruling set} has also been studied in the \CC \cite{BernsHP12,hegeman2014near,cambus2023time,GilibertiP24}.
In this problem, we relax the domination requirement from MIS: every vertex should be within distance two from a node in the independent set.
Very recently, Cambus, Kuhn, Pai and Uitto \cite{cambus2023time} showed this problem has $O(1)$-round algorithm. Deterministic algorithms for 2-ruling set were given by \cite{PaiP22,GilibertiP24}.

\paragraph{Sparse Graphs.}
Arboricity has been intensively studied in \local and \CC
(see, e.g., \cite{BarenboimE10,barenboim2016locality,BarenboimK18,GhaffariS19}).
Since the arboricity is bounded by the average-degree, our results show $O(1)$ round \CC algorithms
for MIS and MM in low arboricity graphs as well.

\paragraph{Small Independence Number.}
While there is no concrete characterization of graphs with small independence number,
random graphs with degree $d$ are known to have independence number $\approx 2n/d \log d$ (see, e.g., \cite{bollobas1976cliques,frieze1990independence} for more precise bounds).
In particular, for random graphs, the Tur\'an bound is nearly-tight and
our algorithm for low independence number ends in $O(1)$ rounds.

\paragraph{Bounded Neighborhood Independence.}
Graphs of bounded neighborhood independence capture a wide range of dense graphs.
This includes line graphs, bounded-growth graphs (e.g., interval graphs \cite{HalldorssonKS03} or disk graphs \cite{HalldorssonK15}).
Graphs with $\beta+1$-bounded neighborhood independence are equivalent to $\beta$-claw-free graphs (i.e., that do not contain $K_{1,\beta+1}$). Such graphs have been extensively studied in structural graph theory (see, e.g., \cite{ChudnovskyS05}).
Motivated by applications to wireless networks,
distributed algorithms for MIS and MM on graphs of bounded independence
have also been studied (see, e.g., \cite{GfellerV07,KuhnWZ08,schneider2010optimal,BarenboimE13,HalldorssonK15}).
More recently, Assadi and Solomon \cite{AssadiS19} gave sublinear algorithms for MIS and MM.
Milenkovi{\'c} and Solomon \cite{milenkovic2020unified} showed that one could compute
$(1+\epsilon)$-matching sparsifiers of size $\BO{n\beta/\epsilon \log(1/\epsilon)}$
on graphs with $\beta$-bounded neighborhood-independence
by simply having each vertex sample adjacent edges.

\paragraph{Organisation of the Paper.}
Notations, definitions and key results from the literature can be found in \cref{sec:prelim}.
\cref{sec:deg-reduction} describes our modified degree reduction algorithm for low average degree graphs.
\cref{sec:spar-independence} provides our algorithms with running time in terms of the neighborhood or global independence numbers.
Our constructions for hard graphs appears in \cref{sec:hard}.

\section{Preliminaries}
\label{sec:prelim}

\paragraph{Graph Theoretic Notation.}
For a graph $H$, let $d(H)$ be its average degree, $\Delta(H)$ its maximum degree, and $n(H)=|V(H)|$ its number of vertices. %
For a vertex $v$, let $N_H(v)$ its set of neighbors, and $\deg_H(v)=|N_H(v)|$ its degree.
The hop-distance between $u$ and $v$ is $\dist_H(u,v)$.
For a subset $S \subseteq V$, let $N_G(S) = \bigcup_{v\in S}N_G(v)$, $\Gamma_G(S)=S\cup N_G(S)$, and for sets $A,B \subseteq V(H)$, let $E_H(A,B) = \{e=(a,b)\in E(H): a\in A, b\in B\}$.

Throughout this paper, the input graph is $G=(V(G),E(G))$, with $n=|V(G)|$ and $m=|E(G)|$.
We denote the independence number of the graph by $\alpha(G)$.
Let $\betag\triangleq\max_{v\in V}\alpha(G[N(v)])$ denote the maximum independence number of an induced neighborhood subgraph.%

We say that a vertex set $A$ is \emph{covered} by an independent set $S$ if $A\subseteq \Gamma(S)$.
The lexicographically first MIS (LFMIS) is defined with respect to an
ordering $v_1, v_2, \ldots, v_n$ of the vertices as the outcome of the following greedy algorithm:
beginning with an empty independent set $S$,
loop through vertices in the given order and add a vertex to $S$ if it is uncovered when reached by the algorithm.
The lexicographically first MM (LFMM) is defined similarly, with respect to adding edges to the matching.

\begin{theorem}[Lenzen's Routing Lemma \cite{lenzen2013optimal}]\label{lemma:routing}
    The following is equivalent to the \CC model:
    In every round each node can send (receive) $\set{b_i}_{i\in[n]}$ bits to (from) the $i$-th node, for any sequence $\set{b_i}_{i\in[n]}$ satisfying $\sum_{i=1}^n b_i=\BO{n\log n}$. In other words, any routing scheme in which no vertex sends or receives more than $\BO{n}$ messages can be performed in $\BO{1}$ rounds.
\end{theorem}

\subsection{Opportunistic Routing}\label{sec:opr}\label{sec:opp route}

We leverage a powerful method of Chechik and Zhang \cite{ST22} to learn $r$-hop neighborhoods in the \CC using opportunistic routing.
A deterministic variant of this method was presented in \cite{Bui2024ImprovedAA}.
\begin{restatable}[Opportunistic Routing, \cite{ST22}]{proposition}{ThmOpR}\label{lemma:op-route}
    Let $c \ge 1$ be a constant.
    Let $G$ be an $n$-vertex graph of maximum degree $\Delta$.
    For any $r$, possibly depending on $n$, such that
    \begin{equation}
        \label{eq:assumption-op-routing}
        \Delta^{2(r+1)} ((r+1)\log \Delta + c\log n) \le n \ ,
    \end{equation}
    there exists a constant round randomized \CC algorithm such that, w.p. $1-n^{-\Theta(c)}$, each vertex learns its $r$-hop neighborhood $E^r(v) = \set{uw\in E, \dist_G(v, u) \le r \text{ and } \dist_G(v, w) \le r}$.
\end{restatable}
For completeness, we include a proof of this result in \Cref{appendix:op-routing}.

As the output of a vertex $v$ in a $r$-round \local algorithm only depends on edges of $E^r(v)$ and identifiers of their endpoints, \cref{lemma:op-route} implies that we can simulate any $r$-round \local algorithm on graphs of maximum degree $\Delta$ provided that \cref{eq:assumption-op-routing} is true.
In this paper, we make use of this observation in the regime described by \cref{thm:small-degree}.
This observation is of independent interest.

\begin{restatable}[Simulating \local Algorithms]{proposition}{ThmLocalSim}\label{thm:small-degree}
    Consider any $r = o(\log n / \log\log n)$-round \local algorithm for a problem with inputs of size $O(\poly(\log n, \Delta))$ per node.
    Then the algorithm can be simulated on \CC in $O(1)$ rounds for graphs of maximum degree $\Delta(G) \le n^{1/(2r+2)}$.
\end{restatable}

By using Ghaffari's algorithm \cite{Ghaffari16} that runs in $O(\log \Delta + \log^3 \log n)$ rounds of \LOCAL \cite{Ghaffari16,GG23} (both for MIS and MM), we immediately obtain from \cref{thm:small-degree} the following fast algorithm for graphs of restricted maximum degree.

\begin{corollary}
    There is a $O(1)$-round \CC algorithm for MIS (MM) on graphs with maximum degree $\Delta = 2^{\sqrt{\log n}/C}$ for some universal constant $C \ge 1$.
    \label{C:low-max-deg}
\end{corollary}

\section{Sparsification via Degree Reduction}
\label{sec:deg-reduction}

The main technical contribution of this section is the following lemma. It is a twist on previous methods that reduce the maximum degree.
We first state the previous result, and then present our modification.
\begin{lemma}[{\cite{ghaffari2018improved} (MIS), \cite{behnezhad2023exponentially} (MM)}]
    There is a $O(1)$-round \CC algorithm that computes an independent set $S$ (matching $M$) such that the residual graph $H = G[V \backslash \Gamma(S)]$ ($H=G[V\setminus M]$) has maximum degree $\Delta(H) = O(\sqrt{\Delta(G)}\log n)$.
    \label{lem:max-deg-reduction}
\end{lemma}

We give a twist to the method by bounding degrees in terms of average degree.
\begin{lemma}
    There is a $O(1)$-round \CC algorithm that computes an independent set $S$ (or a matching $M$) such that the residual graph $H = G[V \backslash \Gamma(S)]$ ($H=G[V\setminus M]$) has maximum degree $\Delta(H) = O(\sqrt{d(G)}\log n)$, w.h.p.
    \label{lem:avg-deg-reduction}
\end{lemma}
The proof for MIS is given below. The proof for MM is in the appendix.

Using this new lemma, combined with \cref{C:low-max-deg}, we obtain one of our contributions.
\begin{corollary}
    There is a $O(1)$-round \CC algorithm for MIS (MM) for graphs $G$ with average degree $d(G) = 2^{O(\sqrt{\log n})}$, w.h.p.
    \label{C:sparse}
\end{corollary}

By applying \cref{lem:avg-deg-reduction} repeatedly, we obtain the following generalization.
We refer readers to \cite{Kon18} or \cite{ghaffari2018improved} for more details on this.
\begin{proposition}
    Let $r \ge 1$. There is an $O(r)$-round randomized \CC algorithm that computes an independent set $S$ (or a matching $M$) in the input graph $G$ such that the residual graph $H\triangleq G[V\setminus\Gamma(S)]$ ($H=G[V\setminus M]$) has maximum degree $\Delta(H) = O(d(G)^{1/2^r} \log^2 n)$, w.h.p.
    \label{prop:repeated-applic}
\end{proposition}

\vspace{-1.3em}
Applying \cref{prop:repeated-applic} with $r = O(\log (\log d(G))/\sl + 1)$, followed by \cref{C:sparse}, we get the second contribution:

\begin{restatable}{theorem}{ThmDeg}\label{thm:degree}
    There is a randomized \CC algorithm for MIS (or MM) running in $O(\log (\log d(G))/\sl + 1)$ rounds, w.h.p.
\end{restatable}

\subsection{Proof of \cref{lem:avg-deg-reduction}}

The following simple $O(1)$-round algorithm has been well studied \cite{ahn2015correlation,Kon18,assadi2018fully,ghaffari2018improved}:

\begin{blackbox}{\OneShotReduceMIS}\label[algorithm]{alg:1 tag}
    \begin{description}
        \item[Input:] A graph $G = (V,E)$, with $n$ vertices and a parameter $p$.
        \item[Output:] An independent set $S$
    \end{description}
    \begin{enumerate}
        \item Sample a subset of vertices $V_F$, where each vertex $v$ is sampled independently w.p. $p$.
        \item Let one designated vertex $v$ collect the edges of the induced subgraph $F\triangleq G[V_F]$.
        \item $v$ locally computes the lexicographically first MIS in $F$. Let $S$ be the resulting independent set.
    \end{enumerate}
\end{blackbox}

We use the following well-known result.
\begin{lemma}[{\cite[Lemma 2.2]{assadi2018fully}}]\label{lemma:deg red restate}
    Let $G=(V,E)$ be a graph and $v$ one of its vertices.
    Let $S\gets\OneShotReduceMIS(G,p)$. Let $G'$ obtained from $G$ by removing
    all edges incident to $\Gamma(S)$. Then, for any $t\ge 0$, we have $\Pr{\deg_{G'}(v) \ge t} \le e^{-tp}$.
\end{lemma}

This lemma implies a variant of our \cref{lem:avg-deg-reduction} in terms of \emph{maximum degree}, by setting $p = 1/\sqrt{\Delta(G)}$. See \cref{claim:w1} for a proof that the call to $\OneShotReduceMIS$ runs in $O(1)$ rounds.

We apply $\OneShotReduceMIS$ with $p = 1/\sqrt{d(G)}$.
Our main task is to bound from above the number of edges in the residual graph $H$.
\begin{lemma}\label{lem:w2}
    Suppose $\Delta(G) = \tilde{O}(n^{1/4})$.
    After running $\OneShotReduceMIS$ with $p = 1/\sqrt{d}$, the graph $F$ contains at most $36n$ edges, w.h.p.
\end{lemma}
\begin{proof}
    For each $e\in E(G)$, let $X_e$ be the random variable indicating if both endpoints of $e$ were sampled in $F$. We have that the probability that $X_e=1$ is $p^2 = 1/d$.
    Let $X = \sum_{e\in E(G)} X_e$ denote the number of edges in $F$. By linearity of expectation, $\Exp{X} = |E(G)|/d = 2n$.
    We show concentration using the method of bounded differences (see \Cref{lemma:bounded difference}). Every vertex joins $F$ independently, and adding/removing a single vertex to $F$ affects the number of edges by at most $\Delta$. Thus, using the assumed bound on $\Delta(G)$, we get
    \begin{align*}
        \Pr{X \geq \Expp{X} + n} \leq \exp\brak{-\frac{n^2}{n \cdot \Delta(G)^2}} \le \exp\parens*{-\tilde{O}(n^{1/2})} \ll 1/\poly(n) \ .
    \end{align*}
\end{proof}

\begin{proof}[Proof of \cref{lem:avg-deg-reduction} for MIS]
    First apply \cref{lem:max-deg-reduction} twice, obtaining a residual graph $G'$ of maximum degree $\Delta(G') = O(\Delta(G)^{1/4}\log n) = \TO{n^{1/4}}$.
    Then run $\OneShotReduceMIS$ on $G'$ with $p=1/\sqrt{d(G')}$.
    By \cref{lem:w2}, the residual graph $F$ contains $O(n)$ edges. We can then gather it at a single node, using Lenzen routing, and complete the MIS $S$ on $G'$ centrally.
    The bound on the maximum degree of the residual graph $H = G[V\setminus \Gamma(S)]$ follows from \cref{lemma:deg red restate} with $t = 1/p \cdot \log n$ and $p=1/\sqrt{d(H)}$, w.h.p.
\end{proof}

\section{Leveraging Independence for Sparsification}
\label{sec:spar-independence}

In this section, we provide fast algorithms for graphs of sublinear neighborhood-independence. We first derive a common sparsification lemma in \cref{ssec:spars-bni}, which is then applied to MIS in \cref{ssec:mis-bni} and to MM in \cref{ssec:mm-bni}.

We first illustrate the proof strategy on graphs with low independence number.

\begin{restatable}{theorem}{ThmP}\label{thm:P}
    Let $G$ be a graph $\alpha(G)= n/d(G)^{\mu}$ for some $\mu=\mu(n) > 0$. There exists a randomized \CC algorithm that, w.h.p, computes an MIS (MM) of $G$ in $\BO{\log(2/\mu) + 1}$ rounds.
\end{restatable}

\begin{proof}
    If $d(G) < 1$, then $G$ has $nd(G)/2 = \BO{n}$ edges and all edges can be gathered by a single node, using Lenzen routing, and complete the MIS $S$ on $G'$ centrally.
    Henceforth, assume $d(G) \ge 1$.
    Apply the algorithm of \cref{prop:repeated-applic} with $r=\log(2/\mu)$, %
    obtaining a residual graph $H$ of maximum degree $\Delta(H) = O(d^{\mu/2}\log^2 n)$, w.h.p.
    By Tur\'an bound and by assumption,
    \begin{align}
        \frac{|V(H)|}{2\Delta(H)}\leq \alpha(H)\leq \alpha(G)\leq \frac{n}{d^{\mu}}\label{eq:4term P}\;.
    \end{align}
    By the degree bound on $H$, we can bound the number of edges in $H$ by
    \begin{align*}
        |E(H)| \le |V(H)| \Delta(H)/2 \leq n \cdot \frac{\Delta(H)^2}{d^\mu} = \TO{n}\ .
    \end{align*}
    We can then compute the MIS (MM) of $H$ in $O(1)$ rounds via \cref{thm:degree}.
\end{proof}

\subsection{Sparsification in Terms of Neighborhood Independence}
\label{ssec:spars-bni}
The goal of this section is to prove \cref{lem:quad-red}. One should think of $t$ as some small $\poly(\log n)$.
\begin{lemma}
    \label{lem:quad-red}
    Let $s\ge 1$ be an integer and $X \subseteq V$ be a subset of vertices such that for each $v \in X$, $\deg_{G[X]}(v) \le s\sqrt{\deg_G(v)}$. Then, $G[X]$ contains $\BO{\betag n\cdot s^2\log n}$ edges.
    \label{L:deg-to-beta}
\end{lemma}

We also use the following standard result on the size of independent sets in $\beta$-bounded neighborhood independence graphs.
\begin{lemma}\label{prop:yi}
    Let $\delta > 0$ and $A \subseteq V_G$ such that $\deg_G(v)\geq \delta$, for all $v \in A$.
    Then, $\alpha(G[A])\le n\betag/\delta$.
\end{lemma}
\begin{proof}
    Let $I$ be an independent set within $A$ of size $\alpha = \alpha(G[A])$. By assumption, each $v\in I$ has at least $\delta$ incident edges (each necessarily incident only on a single node in $I$).
    On the other hand, by definition,
    a vertex can have at most $\betag$ neighbors in $I$. The lemma follows as there are at most $n\betag$ and at least $\alpha \delta$ edges with an endpoint in $I$, so $\alpha \le n\betag/\delta$.
\end{proof}

\begin{proof}[Proof of \cref{L:deg-to-beta}]
    Split the vertices into degree classes as follows: for each $i\in\set{0, 1, \ldots, \log n}$, define:
    \begin{align*}
        d_i = 2^i \ , \quad\text{and}\quad
        V_i  \triangleq\set*{v\in V\mid d_i \le \deg_G(v) < d_{i+1}} \; .
    \end{align*}

    Fix $i\in[\log n]$ and consider the subset $X_i\triangleq X\cap V_i$ of $X$ within $V_i$.
    We have
    \begin{align*}
        \frac{\abs{X_i}}{2s \sqrt{d_i}}\le \alpha(G[X_i])\le \alpha(G[V_i])\le  \frac{n\betag}{d_i}\;.
    \end{align*}
    The first inequality follows by Tur\'an's bound, by the assumed degree bound,
    and the last inequality follows by \cref{prop:yi}.
    Multiply both sides by $(2s\sqrt{d_i})^2$ to get
    \begin{align*}
        \abs{E_{G[X]}(X_i)}\leq \abs{X_i}(2s\sqrt{d_i}) \leq
        \frac{n\betag}{d_i} \cdot (2s\sqrt{d_i})^2 = 4 n\betag \cdot s^2\;.
    \end{align*}
    Therefore, by summing over all degree classes, we get that $G[X]$ contains $\BO{n\betag \cdot s^2\log n}$ edges.
\end{proof}

\subsection{Maximal Independent Set}
\label{ssec:mis-bni}
\begin{theorem}\label{thm:bounded-NI-mis}
    \sloppy{There is a randomized algorithm for MIS in \CC that runs in $\BO{\llbet}$ rounds, \whp.}
\end{theorem}

The uniform degree reduction (\cref{prop:repeated-applic}) does not run any faster than $\Theta(\log\log n)$ rounds on graphs of high average degree.
Hence, we need a different argument for graphs of bounded neighborhood-independence. We show that sampling vertices with probability proportional to the square-root of their degrees yields good sparsification.
By running the degree-reduction algorithm twice, we may assume that all degrees are $\TO{n^{1/4}}$.

\pagebreak
\begin{blackbox}{\OneShotReduceMISNonUnif}\label[algorithm]{alg:mis-claw}
    \begin{description}
        \item[Input:] A graph $G = (V,E)$, with $n$ vertices and $\Delta(G)\leq \TO{n^{1/4}}$.
        \item[Output:] An independent set $S$
    \end{description}
    \begin{enumerate}
        \item Sample a subset of vertices $V_F$, where each vertex $v$ is sampled independently w.p.\ $p_v = 1/\sqrt{\deg_G(v)}$.
              Let a designated vertex $w$ collect the edges of the induced subgraph $F\triangleq G[V_F]$.
        \item The vertex $w$ locally computes the LFMIS $S$ on $F$ and informs each vertex whether it belongs to $S$ or not.
        \item Apply the MIS algorithm of \cref{thm:degree} on the graph $H\triangleq G[V\setminus\Gamma(S)]$.
    \end{enumerate}
\end{blackbox}

\newcommand{\yi}{\alpha(G[V_i])}

Before proving \cref{thm:bounded-NI-mis}, we introduce a refined version of the MIS sparsification lemma (\cref{lemma:deg red restate})
to capture non-uniform sampling.
\begin{lemma}
    \label{lem:mis-sparsification-non-uniform}
    Fix some $n$-vertex graph $G=(V,E)$ and parameters $p_v \in [0,1)$ for each vertex $v\in V$.
    Let $F$ be the random set of vertices obtained by picking each vertex $v\in V$ independently with probability $p_v$.
    Denote by $S$ the independent set obtained by computing the LFMIS in $F$ and let $H = G[V \setminus \Gamma(S)]$ be the residual graph.
    Let $p\in(0,1)$ and $X=\set{u\in V\mid p_u\geq p}$.
    Then, for $v\in V(H)$ and $t \ge 1$, the probability that $v$ has at least $t$ neighbors in $X$ that are within $H$ is bounded by
    \[ \Pr{\deg_{H\cap X}(v) \ge t} \le e^{-tp} \ . \]
\end{lemma}

\begin{proof}
    Consider a vertex $v$.
    The algorithm can be described as running the following martingale
    process. Starting with an empty independent set $S$, iterate over vertices according to a fixed arbitrary order. When we reach a vertex $u \in N(v) \cap X$, if one of its neighbors (earlier in the ordering) already joined $S$, we do nothing; otherwise, we toss a $p_u$-coin and include $u$ in the set $S$ if it succeeds (i.e., w.p.\ $p_u$). We say a vertex is covered if it is in $S$ or it has a neighbor in $S$.

    For each neighbor $u \in X$ that $v$ has in the residual graph, the following must have happened: ``when we reached $u$, it was not covered and we failed the $p_u$-coin''. Let $A_i$ be the event that this event occurred at least $i$ times. This is equivalent to the event $\deg_{H\cap X}(v) \ge i$ that $v$ has residual degree in $X$ of at least $i$. By the definition of conditional probability,
    \[ \Pr{A_{i+1}} = \Pr{A_{i+1} \cap A_i} = \Pr{A_i} \Pr{A_{i+1}|A_i}\ .   \]
    Also, $\Pr{A_{i+1}|A_i} \le 1-p$, as we only toss coins for nodes $u$ in $X$, for which $p_u \ge p$.
    Hence,
    \[ \Pr{A_t} = \Pr{A_0}\prod_{i=1}^{t-1} \Pr{A_{i}|A_{i-1}} \le (1-p)^t \le e^{-tp} \ . \]
\end{proof}

We can now prove \cref{thm:bounded-NI-mis}.

\begin{proof}[Proof of \cref{thm:bounded-NI-mis}]
    We analyze the decrease in degrees after \OneShotReduceMISNonUnif.
    In \OneShotReduceMISNonUnif, all vertices in $X = \cup_{j\le i} V_j$ are sampled with probability at least $1/\sqrt{d_{i+1}}$. Hence, by \cref{lem:mis-sparsification-non-uniform} with $t=c\log n\cdot \sqrt{d_i}$ for some large enough constant $c$, we get that every vertex $v\in H$ has at most $t$ neighbors in $H\cap X$, with probability at least $1-n^{-c}$. Thus, w.h.p, we get $\Delta(H[X])\leq c\log n\cdot \sqrt{d_i}$.

    Next, we explain how to implement \OneShotReduceMISNonUnif in $\BO{1}$ rounds in \CC. For this, we need to show that the induced graph $F$ contains $\BO{n}$ edges with high probability.
    This follows from the assumption that $\Delta < n^{1/3}$ and a straightforward application of the method of bounded differences, see e.g., \cite[Lemma 3]{cambus2023time}.
    By \cref{lem:quad-red} with $s=O(\log n)$, the residual graph $H$ contains $\TO{\betag n}$ edges, and thus we can find an MIS on $H$ in $\BO{\llbet}$ rounds by \Cref{thm:degree}.
\end{proof}

\subsection{Maximal Matching}
\label{ssec:mm-bni}

\begin{theorem}\label{thm:mm-claw-only}
    There exists a randomized \CC algorithm for Maximal Matching that runs in $\RC$ rounds, \whp.
\end{theorem}

We first derive an algorithm for induced subgraphs where all nodes are of roughly the same degree. This is obtained by several rounds of degree reduction, after which the residual subgraph is sparse. We then use this to solve the general case.

Recall that $V_i$ is the set of vertices with degree in the interval $[d_i,d_{i+1})$, where $d_i=2^i$ for $i\in[\ell]$. We denote the $i$-th degree class by $V_i=\set{v\in V\mid d_i\leq \deg_G(v)< d_{i+1}}$.

\begin{proof}
    The algorithm is as follows. We first find a maximal matching $M_i$ on each subgraph $G[V_i]$ (in parallel), and then find a maximal matching on the residual graph $G' = G[V \setminus \cup_i M_i]$.
    To find a maximal matching on $G[V_i]$, we perform two iterations of degree reduction. By \cref{lem:avg-deg-reduction}, the residual subgraph $H_i$ on $G[V_i]$ has maximum degree at most $\BO{\sqrt{d_i}\cdot \log n}$, \whp.
    Then by \cref{L:deg-to-beta}, with $s=\BO{c\log^2 n}=\TO{1}$ each $H_i$ contains $\TO{\betag n}$ edges.
    Thus, the graph $H = \cup_i H_i$, containing only edges within degree classes, has $\TO{\betag n}$ edges. By \cref{thm:degree}, we can find a maximal matching $M$ on $H$ in the desired time complexity.
    For every $i$, let $M_i$ denote the subset of edges in $M$ that are in $G[V_i]$. Then, $M_i$ is maximal, as $M$ is maximal on $H$, and $H$ contains only edges within degree classes.

    In what follows, we find a maximal matching on the residual graph $G' = G[V \setminus \cup_i M_i]$, by showing that $G'$ has $\TO{\betag n}$ edges.
    Let $F_i = V_i \setminus M_i$ be the set of unmatched vertices in $V_i$ after computing a maximal matching on $H$. Observe that $F_i$ is an independent set within $V_i$.
    Thus, by \cref{prop:yi}, $|F_i| \le \alpha(G[V_i]) \le n \betag/d_i$ and the number of edges in $G$ incident to $F_i$ is at most $2 d_i |F_i| \le 2n \betag$.
    Hence, $G' = G[V \setminus \cup_i M_i]$ contains at most $\sum_i 2n \betag = 2n \betag \cdot \log n = \TO{\betag n}$ edges. The round complexity then follows from \cref{thm:degree}.
\end{proof}

\section{Hard Graphs}
\label{sec:hard}

\subsection{MIS}
\newcommand{\Gh}{G_{\mathrm{hard}}}
\newcommand{\Ghmis}{G_{\mathrm{hard}}^{\mathrm{MIS}}}
\newcommand{\Ghmm}{G_{\mathrm{hard}}^{\mathrm{MM}}}
\newcommand{\Vl}{V_{\mathrm{left}}}
\newcommand{\subT}[1]{\mathsf{SubTree}\brak{#1}}
\newcommand{\Anc}{\mathsf{Anc}}
\newcommand{\Th}[1]{\mathcal{T}_h(#1)}
\newcommand{\FF}{\mathcal{F}}
\newcommand{\EE}{\mathcal{E}}
\newcommand{\algA}{\ReduceMis}
\newcommand{\algB}{\ReduceMM}
\newcommand{\algG}{\ReduceMM}
\newcommand{\level}{\mathsf{level}}
\newcommand{\ellval}{\ell/20}
\newcommand{\Ein}{E_{\mathrm{in}}}
\newcommand{\Esamp}{E_{\mathrm{s}}}
\newcommand{\Aj}[1]{L^{\mathrm{out}}_{#1}}
\newcommand{\Bj}[1]{L^{\mathrm{in}}_{#1}}
\newcommand{\prupval}{4\log n/k^{1/(2\cdot 4^r)}}
\newcommand{\prupvall}{\log n/k^{1/(2\cdot 4^r)}}
\newcommand{\spa}{1-2^{-\SL}}

In this section, we analyze the following algorithm that we refer to as $\ReduceMis$. The state-of-the-art algorithm of \cite{ghaffari2018improved} as well as the algorithms we presented here can be seen as special cases of $\ReduceMis$.
The algorithm is as follows. Let $G_0=G$ be an $n$-vertex $m$-edge graph.
In the $i$-th iteration, the algorithm samples a subset $F_i$ of vertices from the current remaining graph $G_i$, where each vertex $v$ is sampled independently with probability $p_{i,v} \leq 1/\sqrt{\delta(G_i)}$, where $\delta(G_i)$ is the minimum degree of $G_i$.
Then, the algorithm computes the lexicographically first MIS $S_i$ on the induced graph $G[F_i]$, and repeats with the graph $G_{i+1} = G[V\setminus \Gamma(S_{\leq i})]$, where $S_{\leq i} = \bigcup_{j\in\set{0,1,\dots,i}}S_j$.
When $G_i$ is empty, the set $S_{\leq i}$ is a maximal independent set and the algorithm ends.

In particular, two special cases of this algorithm are obtained when we set $p_{i,v} = 1/\sqrt{d(G_i)}$ (as in \cite{ghaffari2018improved}) or $p_{i,v} = 1/\sqrt{\deg_{G_i}(v)}$ (as in \cref{prop:repeated-applic}).

We claim that, with high probability, this algorithm needs $\Omega(\log\log n)$ sampling rounds before it is completed.

~\\\textbf{Hard graph distribution:}
Our hard graph distribution is as follows. We take $n= k\times \ell$ vertices and split them into $\ell$ \emph{levels}, with $k$ vertices in each level. In level $i\in \{0,\dots \ell-1\}$, we group vertices into $k^{1-1/4^i}$ \emph{clusters}, which we refer to as the $i$-th level clusters, each containing $k^{1/4^i}$ vertices. In particular, level 0 contains a single cluster of $k$ vertices. We denote the set of vertices in the $i$-th level by $L_i$, and use $L_{\geq i}$ to denote the set of all vertices of level at least $i$.

Next we add edges to the graph.
First, we form a clique out of each cluster.
This will ensure that any independent set in the graph contains at most one vertex from each cluster.
Next, we add random edges to the graph as follows.
Let $u$ be some vertex in the $i$-th level, and let $v$ be some vertex in the $j$-th level, where $i<j$. We add the edge $(u,v)$ with probability $$q_i\triangleq (k^{1/4^i-1})/\log^2 n \ .$$
This completes the description of the hard graph family for MIS, which we denote by $\Ghmis{(k,\ell)}$.

~\\\textbf{Intuition for other graph parameters:}
The reason that this lower bound does not contradict the improved algorithms in terms of $\Delta,d,\betag$, is that in this graph family, we have that $\betag$ and $d(G)$ are both $\Omega(2^{\log^{0.75} n})$, for which
$\log(\log(2^{\log^{0.75} n})/\sqrt{\log n})=
    \log(\log(2^{\log^{0.25} n})) =\Theta(\log\log(n))$.
This holds because
in $G_r$, an active vertex $v$ of level $r$ has $\alpha(G_r[N(v)])\geq k^{1-4/4^{\ell}}q_r/2$ w.h.p, as this is the number of clusters in the last level, with at least one active vertex which is a neighbor of $v$, where $k^{1-4/4^{\ell}}q_r/2\geq k^{1/4^{r+1}}$.

~\\We now state the main theorem of this section.
\begin{theorem}\label{thm:hard graph}
    On the hard graph distribution $\Ghmis(k,\ell)$, where $\ell=\log(\log(k))/20$, the algorithm $\ReduceMis$ takes  $\Omc[\log(\log n)]$ iterations to compute an MIS, with probability at least $1-2^{-\SL}$.
\end{theorem}

We denote the level of a vertex $v$ by $\level(v)$, and the level of cluster $C$ by $\level(C)$.
We emphasize that both iterations and levels employ zero-based indexing.
The key ingredient is the following.
\begin{lemma}\label{lemma:vertices in levels}
    For every $r\in\set{0,1,\ldots,\ell/2}$ and every $r\leq j<\ell$, the graph $G_r$ has at least $k(1-10\ell\cdot r/\log^2n)$ active vertices in $L_j$, with probability at least $\spa$.
\end{lemma}
\Cref{lemma:vertices in levels} implies \Cref{thm:hard graph}
as it shows that after $\ell/2$ iteration at least half of the vertices of level $\ell/2$ are still active, with probability at least $1-2^{-\Theta(\SL)}$, by a union bound, so the algorithm is not yet done.

\begin{proof}[Proof of \Cref{lemma:vertices in levels}]
    We prove the lemma by induction on $r$. Let $\chi = \Theta(2^{-2\sqrt{\log n}} / \ell^3)$, which will bound error probabilities throughout the proof.

    \textbf{Base Case:}
    When $r=0$ every vertex is active and every layer contains exactly $k$ vertices. The induction hypothesis therefore trivially holds.

    \textbf{Induction Step:}
    In the $r$-th iteration, the algorithm samples a random set of vertices $F_r$, where each vertex $v\in V$ joins $F_r$ with probability $p_{r,v}\geq 1/\sqrt{\delta(G_r)}$, and computes an independent set $S_r$ over $G_{r}[F_r]$.
    We will need to use the fact that $p_{r,v}\leq \prupval$, which we deduce
    from the induction hypothesis as follows.
    \begin{claim}\label{claim:pr}
        If the induction hypothesis holds for $r$, then the graph $G_r$ has minimum degree smaller than $\Lambda=\frac{k^{1/4^r}}{16\log^2 n}$ with probability at most $\chi/2$.
    \end{claim}
    \begin{proof}
        We begin with some technical bound on $\Lambda$, which shows that $\Lambda\geq 2^{\log^{3/4}n}$.
        First note that
        \begin{align*}
            \frac{k^{1/4^r}}{2\log^2 n}
            \geq 2^{\log k/4^r-2\log\log n}\;.
        \end{align*}
        Therefore, to prove that $\Lambda\geq 2^{\log^{3/4}n}$, it suffices to show that $\log k/4^r\geq 2\log^{3/4} n$.
        We prove that $\log k/4^r\geq 2\log^{3/4} n$:
        \begin{align*}
            \frac{\log k}{4^r }
            = \frac{\log k}{2^{2r}}
            \geq \frac{\log k}{2^{2\ell}}
            = \frac{\log k}{\log^{1/10} n}
            \geq \frac{\log k}{\log^{1/10} (k^2)}
            \geq \log^{8/10} (k)
            \geq 2\log^{3/4} (k) \;.
        \end{align*}

        ~\\We are now ready to prove the claim.
        The induction hypothesis states that $G_r$ has at least $k(1-10\ell\cdot r/\log^2 n)$ vertices in level $j$ for $j\geq r$.
        Note that $k(1-10\ell\cdot r/\log^2 n)\geq k/2$.

        We split the analysis into two cases: vertices in level at least $r+1$ and vertices in level at most $r$.
        Consider a vertex $u$ in $G_r$ that is level at $h > r$.
        Since every vertex in $L_r$ is a neighbor of $u$ with probability $q_r$, independently of all other vertices, we have that $u$ has in expectation at least $k\cdot q_r/2$ neighbors in $G_r$, where
        \begin{align*}
            \frac{k\cdot q_r}{2}= \frac{k^{1/4^r}}{2\log^2 n}=8\Lambda\;.
        \end{align*}
        Therefore, by a Chernoff bound, $u$ has degree at least $\Lambda$ in $G_r$ with probability at least $1-2^{-\Lambda/12}$.

        Now consider a vertex $u$ in $G_r$ that is in level $h$ for some $h\leq r$.
        For the same reason, $u$ has, in expectation, at least $(\ell-r)\cdot k\cdot q_r/2\geq k\cdot q_r/2$ neighbors in $G_r$, and the same analysis applies.
        To conclude, we get that the minimum degree of $G_r$ is at least $\Lambda$ with probability at least $1-n/2^{-\Lambda/12}$ by a union bound over at most $n$ vertices in $G_r$.
        Since $\Lambda\geq 2^{\log^{3/4}n}$, we have that $1-n/2^{-\Lambda/12}\geq 1-\chi/2$, which completes the proof of the claim.
    \end{proof}

    \newcommand{\fB}{\mathscr{B}}

    Let us call $\fB$ the low-probability (bad) event ``$G_r$ has minimum degree $\delta(G) < \Lambda$'', where by \Cref{claim:pr}, $\Pr{\fB}\leq \chi/2$.
    Conditioning on $\overline{\fB}$ (the complement of $\fB$), every vertex $v$ is sampled in $S_r$ with probability at most \begin{equation}
        \label{eq:prob}
        p_{r,v} \leq 1/\sqrt{\Lambda} = \prupval \ .
    \end{equation}
    However, we emphasize that conditioning on $\fB$ affects the distribution of random edges in $G_r$; in particular, the existence of incident edges in $G_r$ is not independent anymore.

    ~\\In what follows, we explain how to prove the induction step by bounding the number of active vertices that are covered by $S_r$ in each level $j>r$,
    that is, the active vertices in $\Gamma(S_r)\cap L_j$.
    We partition the random set of active vertices in $\Gamma(S_r)\cap L_j$ into two sets, denoted by $\Bj{j}$ and $\Aj{j}$.
    The set $\Bj{j}$ contains all active vertices $u\in L_j$ that are \emph{inside} a sampled cluster (a cluster for which at least one vertex is sampled into $F_r$). That is, $\Bj{j}$ contains all active vertices $u\in L_j$  for which $C(u)\cap F_r\neq\emptyset$, where $C(u)$ is the cluster of $u$.
    The set $\Aj{j}$ contains the rest of the active vertices in $\Gamma(S_r)\cap L_j$. That is, it contains active vertices in $L_j$ that are covered by $S_r$ but are \emph{outside} all sampled clusters.

    We show that \whp $|{\Bj{j}\cup \Aj{j}}|\leq k \cdot (10\ell/\log^2 n)$, and therefore, $\abs{L_j\setminus\Gamma(S_r)}\geq k(1-10\ell/\log^2 n)$ \whp.
    This means that $G_{r+1}$ has at least $k(1-(r+1)\cdot 10\ell/\log^2 n)$
    active vertices in $L_j$, w.h.p, and it therefore completes the induction step.

    \paragraph{Bounding $|{\Bj{j}}|$, the number of vertices whose cluster is sampled.}
    We prove that for every $j>r$, we have that
    $|{\Bj{j}}|\leq k/\log^2 n$, with probability at least $1-2^{-\sqrt{k}} - \chi/2$.
    Note that the randomness for $\Bj{j}$ is over the choice of $F_r$, i.e., of
    the random set of vertices that are sampled into $F_r$.
    So we may condition on $\overline{ \fB }$, thus use \cref{eq:prob}; randomness over edges of $G_r$ play no other role in this part of the argument.
    Let $X_j$ denote the number of clusters in level $j$ that intersect $F_r$.
    Consider some cluster $C$ in level $j$. Let $P_j(C)$ denote the probability that it intersects $F_r$. We have
    \begin{align*}
        P_j(C)=\Pr{F_r\cap C\neq\emptyset~|~\overline{\fB}}
        = 1-(1-p_r)^{\abs{C}}
        \leq p_r\cdot \abs{C} \leq p_r k^{1/4^j}\;.
    \end{align*}

    Therefore,
    $\Exp{X_j~|~\overline{ \fB }}
        =\sum_{C: \level(C)=j}P_j(C)=p_r\cdot k^{1/4^j}k^{1-1/4^j}= k \cdot p_r$.
    By plugging in \cref{eq:prob}, we get that $k\cdot p_r\leq 4\log n\cdot \sqrt{k}$, which means that $\Exp{X_j}\leq 4\log n\cdot k^{1-1/(2\cdot 4^r)}$.
    Let $\gamma=4\log n\cdot k^{1-1/(2\cdot 4^r)}$.

    Each cluster intersects $F_r$ independently of all other clusters, as vertices are sampled into $F_r$ independently of each other.
    This means that $X_j$ is the sum of $k^{1-1/4^j}$ independent Bernoulli random variables, where each one is equal to $1$ with probability at most $p_r\cdot k^{1/4^j}$.
    We can apply Chernoff's inequality on $X_j$ to get that
    \begin{align*}
        \Pr{X_j\geq 6\gamma~|~\overline{ \fB }}\leq 2^{-\gamma}\leq 2^{-\sqrt{k}}\;.
    \end{align*}

    Finally, we have $|{\Bj{j}}|\leq X_j\cdot k^{1/4^j}$ because each cluster in level $j$ contains at most $k^{1/4^j}$ vertices, and hence by recalling that $j>r$ we get
    \begin{align*}
        |{\Bj{j}}|
        \leq 6\gamma k^{1/4^j}
        =
        24 k^{1-1/(2\cdot4^r)}k^{1/4^j}\log n
        \le 24 k^{1-1/4^{r+1}}\log n \leq k/\log^2 n\;,
    \end{align*}
    as needed. By the law of total probability, this holds with probability $\Pr{\fB} + \Pr{X_j < 6\gamma~|~\overline{ \fB }} \geq 1 - 2^{-{\sqrt{k}}} - \chi/2$ as claimed.

    \paragraph{Bounding $|{\Aj{j}}|$, the number of covered vertices whose cluster is not sampled.}
    To bound $|{\Aj{j}}|$, we consider an equivalent process in which we first sample a set $F_r$, sample the random edges in $G[F_r]$ and compute $S_r$, and only then sample the random edges between $F_r$ and $V\setminus F_r$. Thus, we can consider $F_r$ and $S_r$ as fixed and compute $|{\Aj{j}}|$ conditioned on their choice. Recall that $S_r$ must: be a subset of $F_r$, and be an independent set in $G_r[F_r]$.

    It is important to note that the set of random edges that we sample between $F_r$ and $V\setminus F_r$ are only sampled once, as we only sample edges between active vertices where at least one endpoint is also sampled into $F_r$.
    This means that in the next iteration, all of these edges touch at least one inactive endpoint, and therefore are not sampled again.

    Fix $j>r$. We want to bound the number of active vertices from level $j$ that have a neighbor in $S_r$, but which are not in $\Bj{j}$, i.e., their cluster is not sampled.
    For this, we partition the vertices of $S_r$ according to their level. Let $S_r^i$ denote the subset of vertices in $S_r$ of level $i$ for $0\leq i<\ell$.
    We define a random variable $X^i$ equal to the number of active vertices in level $j$ that have a neighbor in $S_r^i$, which are not in $\Bj{j}$.
    Note that $|{\Aj{j}}|=\sum_{0\leq i<\ell}X^i$. For each $i$,
    we have that $X^i$ is a sum of at most $\abs{L_j}$ indicator random variables, for the events that each active vertex in $L_j\setminus\Bj{j}$ has a neighbor in $S_r^i$. These random variables are independent, because after fixing $F_r$ and $S_r$, we add edges between $S_r$ and vertices in $L_j\setminus\Bj{j}$ independently.

    To bound $X^i$, we split into two cases.

    \underline{When $i\leq r$.}
    Any independent set in $G$ that is contained in level $i$ has at most $k^{1-1/4^i}$ vertices -- at most one vertex in each cluster in level $i$. In particular, $|{S_r^i}|\leq k^{1-1/4^i}$.
    Fix some active vertex $v\in L_j\setminus\Bj{j}$.
    Recall that the probability that a random edge is added between vertices in levels $i<j$ is $q_i=k^{1/4^i-1}/\log ^2 n$.
    The probability that $v\in N(S_r^i)$ is thus equal to
    \begin{align*}
        1-(1-q_i)^{|{S_r^i}|}\leq |S_r^i|\cdot q_i \leq 1/\log^2 n\;.
    \end{align*}
    We therefore get that $X^i$ has the distribution of a Binomial random variable with at most $|{L_j}|$ trials and success probability at most $1/\log^2 n$. Therefore, by a Chernoff bound, we have that $\Pr{X^i\geq 6k/\log^2 n}\leq 2^{-6k/\log^2 n}\leq \chi$.
    We emphasize that in this case, we do not use \cref{eq:prob}, thus do not need to condition on $\overline{ \fB }$. And so edges are indeed added in $G_r$ independently, which allows us to use the Chernoff Bound.

    \underline{When $i>r$.} %
    Recall that $X^i$ is the number of active vertices in $L_j\setminus\Bj{j}$, which have a neighbor in $S_r\cap L_i$.
    We define a new random variable $\hat{X}^i$ equal to the number of active vertices in $L_j\setminus\Bj{j}$, which have a neighbor in $F_r \cap L_i$.
    Note that the latter is a larger set, i.e., $X^i\leq \hat{X}^i$. Thus, a bound on $\hat{X}^i$ directly bounds $X^i$.

    Let us first bound the number of vertices in $F_r \cap L_i$. Let $p$ be the upper bound on $p_{r,v}$ given by \cref{eq:prob}. Conditioning on $\overline{\fB}$, a vertex $v$ joins $F_r$ with probability at most $p_{r,v} \leq p$ and $\Exp{F_r \cap L_i~|~\overline{\fB}} \leq p k$. Since sampling in $F_r$ is independent for every vertex (even with the condition on $\overline{\fB}$), a Chernoff bound shows that
    \begin{align*}
        \Pr{ |F_r \cap L_i| > 6k p }
         & \leq \Pr{\fB} + \Pr{ |F_r \cap L_i| > 6 kp ~|~\overline{\fB}} \leq \chi/2 + 2^{-kp} \leq \chi \ .
    \end{align*}
    On the other hand, we can bound the expected maximum degree in $L_j$ of a vertex in $L_i$ by $k q_{r+1}$ since every edge between the two sets is sampled in the input graph with probability at most $\max\set{q_i, q_j} \leq q_{r+1}$. Note that the bound on expected degrees does not rely on \cref{eq:prob}, hence we do \emph{not} need to condition on $\fB$ and edges are added to the input graph independently. And so, by a Chernoff Bound, every vertex in $L_i$ has at most $2 k q_{r+1}$ neighbors in $L_j$ with probability at least $1 - 2^{-kq_i} \geq 1 - \chi$.
    By union bound, with very high probability, both bounds hold and thus the number of vertices in $L_j$ adjacent to a vertex in $L_i \cap F_r$ is at most
    \[
        \hat{X}_i
        \leq 2q \cdot p \cdot k^2
        = 8k^{1 - \frac{1}{2\cdot 4^r} + \frac{1}
                {4^{r+1}}}/\log n \leq k/\log^2 n\ .
        \tag{using $p=4\log n/k^{1/(2\cdot4^r)}$ and $q \leq q_{r+1} = k^{1/4^{r+1}-1}/\log^2 n$, as $r< i, j$.}
    \]

    \paragraph{Putting Things Together.}
    By union bound over pairs $i, j \in \set{0, 1, \ldots, r}$, we have that $|{\Aj{j}}|+|{\Bj{j}}|\leq k \cdot (10\ell/\log^2 n)$ with probability at least $1 - O(\ell^2 \chi)$.
    By the induction hypothesis, $G_{r}$ has at least $k(1-r\cdot 10\ell/\log^2n)$ active vertices in level $j$. In iteration $r$, the set of vertices in level $j$ that become covered is contained in $\Gamma(S_r)\cap L_j$, where $\Gamma(S_r)\cap L_j\subseteq \Aj{j}\cup\Bj{j}$.
    Therefore, the number of active vertices in level $j$ in $G_{r+1}$ is at least $k(1-(r+1)\cdot 10\ell/\log^2n)$ with probability at least $1- O( \ell^2 \chi )$.
    A union bound over all $r/2 \leq \ell$ iterations implies that the induction hypothesis holds in each iteration, with probability at least  $1- O(\ell^3 \chi) \geq 1 - 2^{2^{-\sqrt{\log n}}}$ by definition of $\chi$.
\end{proof}

\subsection{Maximal Matching}
\newcommand{\Er}{E_{\level\leq r}}
\newcommand{\Erb}{E_{\level< r}}
\renewcommand{\Aj}{E^{\mathrm{out}}}
\renewcommand{\Bj}{E^{\mathrm{in}}}
Denote the algorithm for computing a maximal matching specified in \Cref{prop:repeated-applic} by $\algB$.
In this section, we construct a graph, on which the algorithm $\algB$ takes at least $\Omega(\log\log n)$ iterations to compute an MIS with high probability.
This analysis also implies that the state-of-the art algorithm for computing maximal matchings presented in \cite{behnezhad2023exponentially} takes $\Omega(\log\log n)$ iterations to compute a maximal matching on this graph with high probability, and therefore shows that the analysis in \cite{behnezhad2023exponentially} is tight.

Recall how the $\algB$ algorithm works.
Let $G_0=G$ be an $n$-vertex $m$-edge graph.
In the $r$-th iteration, it randomly partitions the active vertices into
$1/p_{r}$ parts, denoted $(F_1^r,\ldots ,F_{1/p_r}^r)$, where $p_r=1/\sqrt{d(G_{r})}$.
The algorithm then computes the lexicographically first maximal matching on each part, denoted by $(M_1^r,\ldots ,M_{1/p_r}^r)$.
The final phase of the algorithm is called a clean-up step, in which it matches vertices of degree at least $\Delta^{0.92}$.
We ignore the implementation details of this phase, as it adds at most $n/\Delta^{0.01}$ edges to the matching with high probability. Therefore, it does not matter exactly how these high-degree vertices are matched as this algorithm matches too few vertices in each iteration to form a maximal matching faster than in $\Omc[\log(\log(n))]$ iterations.

Let $\Esamp^r \triangleq \bigcup_{i\in[1/p_r]}E(G[F_i^r])$ denote the set of edges that are contained in $F_i^r$ for some $i$. We refer to this set as the set of $r$-active edges.
Let $M_{all}^r$ denote the union of all of the matchings $(M_1^r,\ldots ,M_{1/p_r}^r)$ , and let $\calM^r=\bigcup_{j\in\set{0,1,\ldots, r}}M_{all}^j$.
We define $G_{r+1}=G[V\setminus V(\calM^r)]$.
We say that a vertex is not active in iteration $r$ if it is matched before iteration $r$, and say it is active otherwise.

~\\We construct a hard graph distribution, denoted $\Ghmm$.
We take a graph with $n=k\sum_{i=0}^{\ell-1} 4^{i}= k(4^\ell-1)/3$ vertices, instead of $k\ell$, where $\ell=\log\log(k)/20$.
We partition its vertices into $\ell$ levels, where the level $i$ has $4^i \cdot k$ vertices.
This change will help us prove that the average degree does not decrease too much from iteration to iteration.
We partition the vertices of level $i$ into $4^i\cdot k^{1-1/4^i}$ clusters, where each is a clique of size $k^{1/4^i}$.
We also make sure that each cluster has an even number of vertices, which
means the graph has a perfect matching.
This means that every maximal matching matches at least $n/2$ vertices, and that if the current graph contains more than $n/2$ active vertices, then the computes matching is not maximal.
Finally, random edges are added in the same manner as in the MIS proof: Let $u$ be some vertex in the $i$-th level, and let $v$ be some vertex in the $j$-th level, where $i<j$. We add the edge $(u,v)$ with probability $q_i\triangleq (k^{1/4^i-1})/\log^2 n$.
\begin{theorem}\label{thm:mm lower bound}
    On the hard graph distribution $\Ghmm(k,\ell)$, where $\ell=\log(\log(k))/20$,
    the algorithm $\algB$ takes $\Omc[\log(\log n)]$ iterations to compute a maximal matching, with probability at least $\spa$.
\end{theorem}

We use the following lemma.
\begin{lemma}\label{prop:mm lb main}
    For every $r\in\set{0,1,\ldots,\ell/2}$ and every $r\leq j<\ell$, the graph $G_r$ has at least $(4^j\cdot k\cdot \frac{2}{3})\cdot (1-10\ell\cdot r/\log^2n)$ active vertices in $L_j$, with probability at least $\spa$.
\end{lemma}
\Cref{prop:mm lb main} implies that the matching after $\ell/2$ iterations,
the number of active vertices is at least
\begin{align*}
      & \sum_{j=\ell/2}^\ell 4^{j}k\cdot (1-o(1)) \cdot 2/3                         \\
    = & \frac{2k\cdot (1-o(1))}{3} \cdot 4^{\ell/2}\cdot \frac{4^{\ell/2 + 1}-1}{3} \\
    = & \frac{2k\cdot (1-o(1))}{9} \cdot (4^{\ell+1}-4^{\ell/2})                    \\
    = & \frac{8n\cdot (1-o(1))}{9} > 3n/4\;.
\end{align*}
And if the number of active vertices is at least $3n/4$, then the matching is of size smaller $n/2$. Since $G$ contains a perfect matching, the algorithm is not done yet.

\begin{proof}[Proof of \Cref{prop:mm lb main}]
    We prove the lemma by induction on $r$, where the base case $r=0$ is trivial, as $G_0=G$, and then every vertex is active.
    We need the following bound on the sampling probability $p_r$.
    \begin{claim}\label{claim:pr-mm}
        Assuming the induction hypothesis holds for $r-1$, we have that
        \begin{equation}
            p_{r} \leq \log n / k^{1/(2\cdot 4^r)}
            \label{eq:prob-matching}
        \end{equation}
        with probability at least $1-2^{-\SL}$.
    \end{claim}
    \begin{proof}[Proof of \Cref{claim:pr-mm}]
        Recall that $G_{r+1}=G[V\setminus V(\calM^{r})]$, and that $p_{r}=1/\sqrt{d(G_{r})}$.
        The induction hypothesis states that for every $j\geq r$, the graph $G_r$ has at least
        $(4^j\cdot k\cdot \frac{2}{3})\cdot (1-o(1))$ active vertices.
        We look at the number of edges between level $r$ and level $\ell-1$.
        We first count the number of pairs of active vertices with one endpoint in level $r$ and one in level $\ell-1$.
        There are at least $(4^{r} 2k/3)\cdot (4^\ell 2k/3)\geq k^2\cdot 4^{r+4}$ such active pairs between the levels.
        Because we add an edge between each such pair independently with probability $q_r$, we get that the expected number of edges in the graph is at least $k^2\cdot 4^{r+4}\cdot q_r$, and that the graph has at least $k^2\cdot 4^{r+2}\cdot q_i$ edges, with probability at least $1-2^{-k^2\cdot q_i/64}$, with probability at least $1-2^{-\SL}$.
        This means that the average degree of $G_r$ is at least $k^2\cdot 4^{r+4} \cdot q_r/n \geq k\cdot q_r \cdot 4^{r}$, where the inequality follows because $n\leq k\cdot 4^{\ell+1}$.
        We conclude that the average degree is at least $k\cdot q_r/64$, with probability at least $1-2^{-\SL}$.
        Since the average degree is at least $k\cdot q_r/\log n$, we get that $p_{r+1}\leq 1/\sqrt{k\cdot q_r/\log n}\leq \log n/k^{1/(2\cdot 4^r)}$,
        with probability at least $1-2^{-\SL}$.
        Recall that the clean-up step ensures that
        $p_{r+2}\geq (p_{r+1})^{1/0.92}$ w.h.p, which we indeed satisfy.
    \end{proof}
    \newcommand{\fB}{\mathscr{B}}
    Define $\fB$ the low-probability event that the average degree is too large, so that \cref{eq:prob-matching} holds when $\overline{ \fB }$ occurs.
    We have that $\Pr{\fB}\leq 2^{-\SL}$.

    ~\\We now have all the ingredients to prove the \Cref{prop:mm lb main}.
    We prove the induction step for $r+1$ instead of $r$ for brevity.
    We prove that for every $j\geq r+1$ the graph $G_{r+1}$ has at least $(4^j\cdot k\cdot \frac{2}{3})\cdot (1-10\ell\cdot r/\log^2n)$ active vertices in $L_j$.
    To prove the lemma, we bound the number of edges in $\calM^r$, which is useful as $G_{r+1}=G[V\setminus V(\calM^r)]$.
    We use the following notation.
    We define the set $\Er$ as the set of edges with level at most $r$, where we define the level of an edge $e=(u,v)$ as the minimum between the level of $u$ and the level of $v$.
    We split the edges of $\calM^r$ into two sets: the edges in $\Er$ and other edges.

    \paragraph{Bounding $\abs{\Er\cap \calM^r}$. } We claim that $\abs{\Er\cap \calM^r}\leq \frac{4^{r+1}-1}{3}\cdot k$. This is because the number of vertices with level at most $r$ is at most
    \begin{align*}
        \sum_{i=0}^{r} 4^ik = \frac{4^{r+1}-1}{3}\cdot k\;.
    \end{align*}
    This means that the set of edges $\abs{\Er\cap \calM^r}$ match at most $2\cdot 4^{r+1}\cdot k/3$ vertices in total, where at least half of the vertices are of level at most $r$, otherwise the level of the edges is not at most $r$.
    Therefore, for every $j\geq r+1$, among the $4^j\cdot k$ vertices in level $j$, at most a third of them are matched in $\Er\cap \calM^r$, and the rest are active in $G_{r+1}$.

    ~\\In what follows, we bound the number of remaining edges in $\calM^r$, i.e., those in $\calM^r\setminus\Er$.
    We prove that this set of edges contains at most $\tilde{\mathcal{O}}(n/\SL) \ll k/\log n$ edges, with probability at least $1-2^{-\SL}$.
    This proves that there are at least $(4^j\cdot k\cdot \frac{2}{3})\cdot (1-o(1))$ active vertices in $L_j$ in $G_{r+1}$, for every $j\geq r+1$, as the set $\calM^r\setminus\Er$ can deactivate at most $2k/\log n=o(k)$ vertices in each level.
    This completes the proof, as we show that the induction hypothesis holds for $r+1$.

    ~\\For the rest of this proof, we focus on bounding the number of edges in $\calM^r\setminus\Er$.
    Recall that we say that an edge is active in iteration $i$, if both of its endpoints fall inside the same part $F_{h}^i$ for some $h$, and
    denote by $\Esamp^i$ the set of active edges in iteration $i$.
    Let $\bigcup_{i\leq r}\Esamp^i$ denote the set of edges that are active in all iterations $i\leq r$. Let $E'=\bigcup_{i\leq r}\Esamp^i\setminus\Er$ denote the set of edges of level at least $r+1$ that are active in some iteration $i\leq r$.
    We partition the set of edges $E'$ into two sets of edges $\Bj$ and $\Aj$.
    The set $\Bj$ contains all edges $e\in E'$ for which
    both endpoints belong to the same cluster.
    The set $\Aj$ contains the rest of the edges in $E'$.
    We claim that $\calM^r \setminus\Er\subseteq E'$, which follows as an edge can join the matching $\calM^r$ only if it is active in some iteration $i\leq r$, and therefore $\calM^r \setminus\Er\subseteq E'=\Bj\cup\Aj$.
    To bound the left-hand side, we bound the size of $\Bj$ and $\Aj$, separately.

    ~\\We show that $\Bj$ contains at most $2k/\log n$ edges with probability at least $1-2^{-\SL}$.

    \paragraph{Bounding $|{\Bj}|$, the number of in-cluster edges of level at least $r+1$.}
    For $j\geq r+1$,
    let $Z_j$ be a random variable equal to $|\Bj\cap \Esamp^r|$, which is the number of active edges in iteration $r$ which belong to a cluster of level $j$.
    The probability of an edge becoming active in iteration $r$ is significantly larger than in previous iterations.
    Therefore, it suffices to bound the size of the set $|\Bj\cap \Esamp^r|$, as
    $\sum_{i=0}^{r}|\Bj\cap \Esamp^i|$ is dominated by $r\cdot |\Bj\cap \Esamp^r|$.

    Henceforth, we shall assume $\overline{ \fB }$ holds, so we case use \cref{eq:prob-matching}. Note that the argument here does not rely on the randomness of edges in $G_{r+1}$.

    We compute $\Exp{Z_j~|~\overline{\fB}}$ as follows.
    For every edge $e$ inside a cluster of level $j$, let $X_e$ denote the indicator random variable for the event that $e$ is active in iteration $r$.
    Then, $Z_j$ is the sum of $m_j$ indicator random variables, where
    $m_j\leq 4^j\cdot k^{1-1/4^j}\cdot (k^{1/4^j})^2\leq  4^j\cdot k^{1+1/4^j}$. For each $X_e$ the probability that $X_e=1$ is $p_{r}$, where $p_{r}\leq \log n/k^{1/{2\cdot 4^r}}$.
    Therefore,
    \begin{align*}
        \Exp{Z_j~|~\overline{\fB}}
        =m_j\cdot p_r
        \leq 4^j\cdot k^{1+1/4^j-1/4^{r}} \cdot \log n
        \leq 4^j\log n \cdot k^{1-3/4^{r+1}}\;,
    \end{align*}
    where the last inequality follows because $r<j$.

    We can also present $Z_j$ as the sum of different \emph{independent} random variables, to prove a concentration bound on $Z_j$.
    For every cluster $C$ with $\level(C)=j$, let $X_C$ denote the number of active edges in $C$.
    Note that the random variables $\set{X_C}_{C:\;\level(C)=j}$ are independent (even subject to the conditioning on $\overline{\fB}$). That is because activation of an edge depends only on the part $F_x^r$ that its endpoint belong to and each vertex samples its part independently.
    A cluster in level $j$ contains at most $k^{1/4^j}$ vertices, therefore $X_C\leq (k^{1/4^j})^2$ (always).
    Let $Y_j=\sum_{C:\; \level(C)=j} X_C$.
    We have $Z_j=Y_j$ as they count the same quantity, and therefore $\Exp{Z_j~|~\overline{\fB}}=\Exp{Y_j~|~\overline{\fB}}$.
    We apply Hoeffding's inequality on $Y_j$, as it is the sum of $4^j\cdot k^{1-1/4^j}$ independent random variables in the interval $[0,(k^{1/4^j})^2]$.
    Therefore, we get
    \begin{align*}
        \Pr{Y\geq \Exp{Y} + t~|~\overline{\fB}}
        \leq\exp(-2t^2/ (4^j\cdot k^{1-1/4^j}\cdot (k^{1/4^j})^4))
        \leq\exp(-2t^2/ (4^j\cdot k^{1+3/4^j}))\;.
    \end{align*}
    We plug in $t=n/2^{2\sl}$, where $t\geq 2\Exp{Y}$ and get that
    \begin{align*}
        \Pr{Y\geq \Exp{Y} + t}\leq \Pr{\fB} + \exp(-n^{1/4}/2^{4\sl})\;.
    \end{align*}
    Therefore, \whp we get that $Y_j\leq 2n/2^{2\sl}$. This means that $|\Bj|\leq 2\ell\cdot n/2^{2\sl}$.
    This completes the bound on $|\Bj|$.

    ~\\ In what follows, we bound the number of edges in $\Aj$ that joins the matching $\calM^r$, we denote this quantity by $|\Aj\cap \calM^r|$.

    \paragraph{Bounding $|\Aj\cap \calM^r|$, the number of out-cluster edges of level at least $r+1$, that join the matching $\calM^r$.}
    Let $U$ denote the set of vertices of level at least $r+1$ that are active in iteration $r$, i.e. $U=\bigcup_{j\geq r+1}L_j \setminus V(\calM^r)$.
    For every vertex $v\in U$ we define a random variable $Z_v$.
    Informally, this is an indicator random variable for the event that $v$ ``can'' add an edge $e\in \Aj$ to $\calM^r$.
    Formally, $Z_v$ is defined as follows.
    Given a partition of the vertices into parts $(F_1^r,\ldots,F_{1/p_r}^r)$,
    where a vertex $v$ is in part $F_h^r$ for some $h\in[1/p_r]$, we define $Z_v$ as the indicator random variable for the event that $v$ has a neighbor $u\in U-C(v)$ which is also in $F_h^r$. Since $u$ and $v$ are neighbors, and they belong to different clusters, we have that $u$ and $v$ are not in the same level.

    Let $Z=\sum_{v\in U}Z_v$, i.e., $Z$ is the sum of all $Z_v$ over all active vertices of level at least $r+1$.
    We claim that $Z\geq |\Aj\cap \calM^r|$, which follows as each vertex $v\in U$ can add at most one edge to $\calM^r\cap \Aj$, and if $Z_v=0$, then $v$ does not add any edge to $\calM^r$.

    We are left with bounding $Z$.
    Henceforth, we shall assume $\overline{ \fB }$ holds, so we case use \cref{eq:prob-matching}.
    We sample a partition of the vertices into parts $(F_1^r,\ldots,F_{1/p_r}^r)$, and then we sample the edges between the vertices in $U$.
    This means that the random variables $\set{Z_v}_{v\in U}$ are independent.
    More precisely, these variables are not independent in general because any one of these variables reveals information about the size of a part in the partition, which then restricts the other variable. Yet, after conditioning on the partition, they become independent, which is sufficient for our needs.

    We fix a vertex $v\in U$ and show that $\Exp{Z_v~|~\overline{\fB}}\leq 1/k^{1/4^r}$.
    Assume without loss of generality that $v$ is in the part $F_1^r$.
    The expected number of vertices in $F_1^r$ is $\abs{U}\cdot p_r$.
    Let $N_v$ denote the number of neighbors of $v$ in $F_1^r$ that are not in the same cluster (or level) as $v$.
    Note that every vertex in $F_1^r$ is a neighbor of $v$ independently with probability $q_{\level(v)}\leq q_{r+1}$.
    Then,
    \begin{align*}
        \Exp{N_v~|~\overline{\fB}}\leq \Exp{q_{r+1}\cdot \abs{F_1^r}~|~\overline{\fB}}\leq np_r\cdot q_{r+1}\;.
    \end{align*}
    We plug in the value of $p_r\leq \log n/k^{1/(2\cdot 4^r)}$ from \cref{eq:prob-matching} and $q_{r+1}=k^{1/4^{r+1}-1}/\log^2 n $. Let $\beta=k^{1/4^r}$, then we have that $p_r\leq \log n / \sqrt{\beta}$, and $q_{r+1}=\beta^{1/4}/(k\cdot \log^2 n )$, and therefore,
    \begin{align*}
        n\cdot p_r\cdot q_{r+1}
        = \frac{k(4^\ell-1)}{3}\cdot p_r\cdot q_{r+1}
        = \frac{4^\ell-1}{3}\cdot \frac{k \cdot \log n \cdot \beta^{1/4}}{k\cdot \log^2 n\cdot \sqrt{\beta}}
        = \frac{4^\ell-1}{3\log n\cdot \beta^{1/4}}
        \leq 1/\beta^{1/4} \;.
    \end{align*}
    Therefore, we have that $\Exp{N_v~|~\overline{\fB}}\leq 1/\beta^{1/4}$, and since $Z_v$ is the indicator random variable for the event that $\set{N_v>0}$, we get that $\Exp{Z_v~|~\overline{\fB}}\leq 1/\beta^{1/4}$.
    By Markov's inequality, we get that $\Pr{Z_v>0~|~\overline{\fB}}\leq \Exp{Z_v~|~\overline{\fB}}\leq 1/\beta^{1/4}$.

    Next, we apply a Chernoff bound on $Z$, which is a Binomial random variable with $\abs{U}$ trials and success probability at most $1/\beta^{1/4}$.
    Therefore, the expectation of $Z$ is at most $n/k^{1/4^r}\ll n/\SL$, and therefore by a Chernoff bound, we get that $Z\leq 6n/\SL$ with probability at least $1-2^{-n/\SL}+ \Pr{\fB}$.
    This proves that
    \begin{align*}
        \Pr{|\Aj\cap \calM^r|> 6n/\SL}
        \leq \Pr{\fB} + \Pr{|\Aj\cap \calM^r|> 6n/\SL~|~\overline{\fB}}
        \leq 2\cdot 2^{-n/\SL}\;,
    \end{align*}
    which completes this part of the proof.

    \paragraph{Putting It All Together.}
    Fix some level $j\geq r+1$.
    We showed that $\Er\cap \calM^r$ covers at most a fraction of $2/3$ of the vertices at level $j$.
    For the remaining edges in $\calM^r$, i.e., $\calM^r\setminus\Er$, we defined the sets $\Bj$ and $\Aj$ and showed that each set contains at most $\BO{2n/2^{\SL}}$ and $6n/\SL$ edges, respectively.
    This means $\calM^r\setminus\Er$ contains at most $o(k)$ edges with high probability.
    This means that at least $2(1-o(1))/3$ fraction of the vertices in level at $j\geq r+1$ are active in $G_{r+1}$, which is what we wanted to prove.
    This completes the proof of the induction step, and therefore the proof of \Cref{prop:mm lb main}.
\end{proof}

\bibliographystyle{alpha}
\bibliography{refs}

\newcommand{\etalchar}[1]{$^{#1}$}
\begin{thebibliography}{GGK{\etalchar{+}}18}

\bibitem[ABB{\etalchar{+}}19]{AssadiBBMS19}
Sepehr Assadi, MohammadHossein Bateni, Aaron Bernstein, Vahab~S. Mirrokni, and
  Cliff Stein.
\newblock Coresets meet {EDCS:} algorithms for matching and vertex cover on
  massive graphs.
\newblock In {\em Proceedings of the Thirtieth Annual {ACM-SIAM} Symposium on
  Discrete Algorithms, {SODA} 2019, San Diego, California, USA, January 6-9,
  2019}, pages 1616--1635. {SIAM}, 2019.

\bibitem[ACG{\etalchar{+}}15]{ahn2015correlation}
KookJin Ahn, Graham Cormode, Sudipto Guha, Andrew McGregor, and Anthony Wirth.
\newblock Correlation clustering in data streams.
\newblock In {\em International Conference on Machine Learning}, pages
  2237--2246. PMLR, 2015.

\bibitem[AKNS24]{Assadi0NS24}
Sepehr Assadi, Christian Konrad, Kheeran~K. Naidu, and Janani Sundaresan.
\newblock O(log log n) passes is optimal for semi-streaming maximal independent
  set.
\newblock In Bojan Mohar, Igor Shinkar, and Ryan O'Donnell, editors, {\em
  Proceedings of the 56th Annual {ACM} Symposium on Theory of Computing, {STOC}
  2024, Vancouver, BC, Canada, June 24-28, 2024}, pages 847--858. {ACM}, 2024.

\bibitem[AKZ22]{AK022}
Sepehr Assadi, Gillat Kol, and Zhijun Zhang.
\newblock Rounds vs communication tradeoffs for maximal independent sets.
\newblock In {\em 63rd {IEEE} Annual Symposium on Foundations of Computer
  Science, {FOCS} 2022, Denver, CO, USA, October 31 - November 3, 2022}, pages
  1193--1204. {IEEE}, 2022.

\bibitem[AOSS18]{assadi2018fully}
Sepehr Assadi, Krzysztof Onak, Baruch Schieber, and Shay Solomon.
\newblock Fully dynamic maximal independent set with sublinear in n update
  time.
\newblock In {\em Proceedings of the 50th Annual ACM SIGACT Symposium on theory
  of computing}, pages 815--826, 2018.

\bibitem[AS19]{AssadiS19}
Sepehr Assadi and Shay Solomon.
\newblock When algorithms for maximal independent set and maximal matching run
  in sublinear time.
\newblock In {\em 46th International Colloquium on Automata, Languages, and
  Programming, {ICALP} 2019, July 9-12, 2019, Patras, Greece}, volume 132 of
  {\em LIPIcs}, pages 17:1--17:17. Schloss Dagstuhl - Leibniz-Zentrum f{\"{u}}r
  Informatik, 2019.

\bibitem[BBH{\etalchar{+}}21]{balliu2021lower}
Alkida Balliu, Sebastian Brandt, Juho Hirvonen, Dennis Olivetti, Mika{\"e}l
  Rabie, and Jukka Suomela.
\newblock Lower bounds for maximal matchings and maximal independent sets.
\newblock {\em Journal of the ACM (JACM)}, 68(5):1--30, 2021.

\bibitem[BBKO22]{BBKO22stoc}
Alkida Balliu, Sebastian Brandt, Fabian Kuhn, and Dennis Olivetti.
\newblock Distributed $\delta$-coloring plays hide-and-seek.
\newblock In Stefano Leonardi and Anupam Gupta, editors, {\em {STOC} '22: 54th
  Annual {ACM} {SIGACT} Symposium on Theory of Computing, Rome, Italy, June 20
  - 24, 2022}, pages 464--477. {ACM}, 2022.

\bibitem[BBKO23]{balliu2023distributed}
Alkida Balliu, Sebastian Brandt, Fabian Kuhn, and Dennis Olivetti.
\newblock Distributed maximal matching and maximal independent set on
  hypergraphs.
\newblock In {\em Proceedings of the 2023 Annual ACM-SIAM Symposium on Discrete
  Algorithms (SODA)}, pages 2632--2676. SIAM, 2023.

\bibitem[BBO22]{BBKO22sicomp}
Alkida Balliu, Sebastian Brandt, and Dennis Olivetti.
\newblock Distributed lower bounds for ruling sets.
\newblock {\em {SIAM} J. Comput.}, 51(1):70--115, 2022.

\bibitem[BCC{\etalchar{+}}24]{Bui2024ImprovedAA}
Hong~Duc Bui, Shashwat Chandra, Yi-Jun Chang, Michal Dory, and Dean
  Leitersdorf.
\newblock Improved all-pairs approximate shortest paths in congested clique.
\newblock {\em Proceedings of the 43rd ACM Symposium on Principles of
  Distributed Computing}, 2024.

\bibitem[BDH{\etalchar{+}}19]{behnezhad2019fully}
Soheil Behnezhad, Mahsa Derakhshan, MohammadTaghi Hajiaghayi, Cliff Stein, and
  Madhu Sudan.
\newblock Fully dynamic maximal independent set with polylogarithmic update
  time.
\newblock In {\em 2019 IEEE 60th Annual Symposium on Foundations of Computer
  Science (FOCS)}, pages 382--405. IEEE, 2019.

\bibitem[BE76]{bollobas1976cliques}
B{\'e}la Bollob{\'a}s and Paul Erd{\"o}s.
\newblock Cliques in random graphs.
\newblock In {\em Mathematical Proceedings of the Cambridge Philosophical
  Society}, volume~80, pages 419--427. Cambridge University Press, 1976.

\bibitem[BE10]{BarenboimE10}
Leonid Barenboim and Michael Elkin.
\newblock Sublogarithmic distributed {MIS} algorithm for sparse graphs using
  nash-williams decomposition.
\newblock {\em Distributed Comput.}, 22(5-6):363--379, 2010.

\bibitem[BE13]{BarenboimE13}
Leonid Barenboim and Michael Elkin.
\newblock Distributed deterministic edge coloring using bounded neighborhood
  independence.
\newblock {\em Distributed Comput.}, 26(5-6):273--287, 2013.

\bibitem[BEPS16]{barenboim2016locality}
Leonid Barenboim, Michael Elkin, Seth Pettie, and Johannes Schneider.
\newblock The locality of distributed symmetry breaking.
\newblock {\em Journal of the ACM (JACM)}, 63(3):1--45, 2016.

\bibitem[BHH23]{behnezhad2023exponentially}
Soheil Behnezhad, Mohammadtaghi Hajiaghayi, and David~G Harris.
\newblock Exponentially faster massively parallel maximal matching.
\newblock {\em Journal of the ACM}, 70(5):1--18, 2023.

\bibitem[BHP12]{BernsHP12}
Andrew Berns, James Hegeman, and Sriram~V. Pemmaraju.
\newblock Super-fast distributed algorithms for metric facility location.
\newblock In {\em Automata, Languages, and Programming - 39th International
  Colloquium, {ICALP} 2012, Warwick, UK, July 9-13, 2012, Proceedings, Part
  {II}}, volume 7392 of {\em Lecture Notes in Computer Science}, pages
  428--439. Springer, 2012.

\bibitem[BK18]{BarenboimK18}
Leonid Barenboim and Victor Khazanov.
\newblock Distributed symmetry-breaking algorithms for congested cliques.
\newblock In Fedor~V. Fomin and Vladimir~V. Podolskii, editors, {\em Computer
  Science - Theory and Applications - 13th International Computer Science
  Symposium in Russia, {CSR} 2018, Moscow, Russia, June 6-10, 2018,
  Proceedings}, volume 10846 of {\em Lecture Notes in Computer Science}, pages
  41--52. Springer, 2018.

\bibitem[CFG{\etalchar{+}}19]{CFGUZ19}
Yi{-}Jun Chang, Manuela Fischer, Mohsen Ghaffari, Jara Uitto, and Yufan Zheng.
\newblock The complexity of ({\(\Delta\)}+1) coloring in congested clique,
  massively parallel computation, and centralized local computation.
\newblock pages 471--480. {ACM}, 2019.

\bibitem[CKPU23]{cambus2023time}
M{\'e}lanie Cambus, Fabian Kuhn, Shreyas Pai, and Jara Uitto.
\newblock Time and space optimal massively parallel algorithm for the 2-ruling
  set problem.
\newblock {\em arXiv preprint arXiv:2306.00432}, 2023.

\bibitem[CS05]{ChudnovskyS05}
Maria Chudnovsky and Paul~D. Seymour.
\newblock The structure of claw-free graphs.
\newblock In {\em Surveys in Combinatorics, 2005: Invited lectures from the
  Twentieth British Combinatorial Conference, Durham, UK, July 2005}, volume
  327 of {\em London Mathematical Society Lecture Note Series}, pages 153--171.
  Cambridge University Press, 2005.

\bibitem[CZ22]{ST22}
Shiri Chechik and Tianyi Zhang.
\newblock Constant-round near-optimal spanners in congested clique.
\newblock In {\em Proceedings of the 2022 ACM Symposium on Principles of
  Distributed Computing}, PODC'22, page 325–334. Association for Computing
  Machinery, 2022.

\bibitem[DKO14]{DKO13}
Andrew Drucker, Fabian Kuhn, and Rotem Oshman.
\newblock On the power of the congested clique model.
\newblock In {\em {ACM} Symposium on Principles of Distributed Computing,
  {PODC} '14, Paris, France, July 15-18, 2014}, pages 367--376. {ACM}, 2014.

\bibitem[DP09]{dubhashi2009concentration}
Devdatt~P Dubhashi and Alessandro Panconesi.
\newblock {\em Concentration of measure for the analysis of randomized
  algorithms}.
\newblock Cambridge University Press, 2009.

\bibitem[FGG{\etalchar{+}}23]{FGGKR23}
Salwa Faour, Mohsen Ghaffari, Christoph Grunau, Fabian Kuhn, and V{\'{a}}clav
  Rozhon.
\newblock Local distributed rounding: Generalized to mis, matching, set cover,
  and beyond.
\newblock In {\em Proceedings of the 2023 {ACM-SIAM} Symposium on Discrete
  Algorithms, {SODA} 2023, Florence, Italy, January 22-25, 2023}, pages
  4409--4447. {SIAM}, 2023.

\bibitem[Fis20]{fischer2020improved}
Manuela Fischer.
\newblock Improved deterministic distributed matching via rounding.
\newblock {\em Distributed Computing}, 33(3-4):279--291, 2020.

\bibitem[FMU22]{FMU22}
Manuela Fischer, Slobodan Mitrovic, and Jara Uitto.
\newblock Deterministic (1+\emph{{\(\epsilon\)}})-approximate maximum matching
  with poly(1/\emph{{\(\epsilon\)}}) passes in the semi-streaming model and
  beyond.
\newblock In {\em {STOC} '22: 54th Annual {ACM} {SIGACT} Symposium on Theory of
  Computing, Rome, Italy, June 20 - 24, 2022}, pages 248--260. {ACM}, 2022.

\bibitem[Fri90]{frieze1990independence}
Alan~M Frieze.
\newblock On the independence number of random graphs.
\newblock {\em Discrete Mathematics}, 81(2):171--175, 1990.

\bibitem[GG23]{GG23}
Mohsen Ghaffari and Christoph Grunau.
\newblock Faster deterministic distributed {MIS} and approximate matching.
\newblock In {\em Proceedings of the 55th Annual {ACM} Symposium on Theory of
  Computing, {STOC} 2023, Orlando, FL, USA, June 20-23, 2023}, pages
  1777--1790. {ACM}, 2023.

\bibitem[GGH{\etalchar{+}}23]{GHIR23}
Mohsen Ghaffari, Christoph Grunau, Bernhard Haeupler, Saeed Ilchi, and
  V{\'{a}}clav Rozhon.
\newblock Improved distributed network decomposition, hitting sets, and
  spanners, via derandomization.
\newblock In {\em Proceedings of the 2023 {ACM-SIAM} Symposium on Discrete
  Algorithms, {SODA} 2023, Florence, Italy, January 22-25, 2023}, pages
  2532--2566. {SIAM}, 2023.

\bibitem[GGK{\etalchar{+}}18]{ghaffari2018improved}
Mohsen Ghaffari, Themis Gouleakis, Christian Konrad, Slobodan Mitrovi{\'c}, and
  Ronitt Rubinfeld.
\newblock Improved massively parallel computation algorithms for mis, matching,
  and vertex cover.
\newblock {\em arXiv preprint arXiv:1802.08237}, 2018.

\bibitem[GGR21]{ghaffari2021improved}
Mohsen Ghaffari, Christoph Grunau, and V{\'a}clav Rozho{\v{n}}.
\newblock Improved deterministic network decomposition.
\newblock In {\em Proceedings of the 2021 ACM-SIAM Symposium on Discrete
  Algorithms (SODA)}, pages 2904--2923. SIAM, 2021.

\bibitem[Gha16]{Ghaffari16}
Mohsen Ghaffari.
\newblock An improved distributed algorithm for maximal independent set.
\newblock In {\em Proceedings of the Twenty-Seventh Annual {ACM-SIAM} Symposium
  on Discrete Algorithms, {SODA} 2016}. {SIAM}, 2016.

\bibitem[Gha17]{G17}
Mohsen Ghaffari.
\newblock Distributed {MIS} via all-to-all communication.
\newblock In Elad~Michael Schiller and Alexander~A. Schwarzmann, editors, {\em
  Proceedings of the {ACM} Symposium on Principles of Distributed Computing,
  {PODC} 2017, Washington, DC, USA, July 25-27, 2017}, pages 141--149. {ACM},
  2017.

\bibitem[GP24]{GilibertiP24}
Jeff Giliberti and Zahra Parsaeian.
\newblock Massively parallel ruling set made deterministic.
\newblock In Dan Alistarh, editor, {\em 38th International Symposium on
  Distributed Computing, {DISC} 2024, October 28 to November 1, 2024, Madrid,
  Spain}, volume 319 of {\em LIPIcs}, pages 29:1--29:21. Schloss Dagstuhl -
  Leibniz-Zentrum f{\"{u}}r Informatik, 2024.

\bibitem[GS19]{GhaffariS19}
Mohsen Ghaffari and Ali Sayyadi.
\newblock Distributed arboricity-dependent graph coloring via all-to-all
  communication.
\newblock In Christel Baier, Ioannis Chatzigiannakis, Paola Flocchini, and
  Stefano Leonardi, editors, {\em 46th International Colloquium on Automata,
  Languages, and Programming, {ICALP} 2019, July 9-12, 2019, Patras, Greece},
  volume 132 of {\em LIPIcs}, pages 142:1--142:14. Schloss Dagstuhl -
  Leibniz-Zentrum f{\"{u}}r Informatik, 2019.

\bibitem[GV07]{GfellerV07}
Beat Gfeller and Elias Vicari.
\newblock A randomized distributed algorithm for the maximal independent set
  problem in growth-bounded graphs.
\newblock In {\em Proceedings of the Twenty-Sixth Annual {ACM} Symposium on
  Principles of Distributed Computing, {PODC} 2007, Portland, Oregon, USA,
  August 12-15, 2007}, pages 53--60. {ACM}, 2007.

\bibitem[HK15]{HalldorssonK15}
Magn{\'{u}}s~M. Halld{\'{o}}rsson and Christian Konrad.
\newblock Distributed large independent sets in one round on
  bounded-independence graphs.
\newblock In {\em Distributed Computing - 29th International Symposium, {DISC}
  2015, Tokyo, Japan, October 7-9, 2015, Proceedings}, volume 9363 of {\em
  Lecture Notes in Computer Science}, pages 559--572. Springer, 2015.

\bibitem[HKS03]{HalldorssonKS03}
Magn{\'{u}}s~M. Halld{\'{o}}rsson, Guy Kortsarz, and Hadas Shachnai.
\newblock Sum coloring interval and k-claw free graphs with application to
  scheduling dependent jobs.
\newblock {\em Algorithmica}, 37(3):187--209, 2003.

\bibitem[HPS14]{hegeman2014near}
James~W Hegeman, Sriram~V Pemmaraju, and Vivek~B Sardeshmukh.
\newblock Near-constant-time distributed algorithms on a congested clique.
\newblock In {\em International Symposium on Distributed Computing}, pages
  514--530. Springer, 2014.

\bibitem[JN18]{JN18}
Tomasz Jurdzinski and Krzysztof Nowicki.
\newblock {MST} in \emph{O}(1) rounds of congested clique.
\newblock In Artur Czumaj, editor, {\em Proceedings of the Twenty-Ninth Annual
  {ACM-SIAM} Symposium on Discrete Algorithms, {SODA} 2018, New Orleans, LA,
  USA, January 7-10, 2018}, pages 2620--2632. {SIAM}, 2018.

\bibitem[KMW10]{kuhn2010local}
Fabian Kuhn, Thomas Moscibroda, and Roger Wattenhofer.
\newblock Local computation: Lower and upper bounds.
\newblock {\em arXiv preprint arXiv:1011.5470}, 2010.

\bibitem[KMW16]{KMW16}
Fabian Kuhn, Thomas Moscibroda, and Roger Wattenhofer.
\newblock Local computation: Lower and upper bounds.
\newblock {\em J. {ACM}}, 63(2):17:1--17:44, 2016.

\bibitem[Kon18]{Kon18}
Christian Konrad.
\newblock {MIS} in the congested clique model in o(log log {\(\Delta\)})
  rounds.
\newblock {\em CoRR}, abs/1802.07647, 2018.

\bibitem[KWZ08]{KuhnWZ08}
Fabian Kuhn, Roger Wattenhofer, and Aaron Zollinger.
\newblock Ad hoc networks beyond unit disk graphs.
\newblock {\em Wirel. Networks}, 14(5):715--729, 2008.

\bibitem[Len13]{lenzen2013optimal}
Christoph Lenzen.
\newblock Optimal deterministic routing and sorting on the congested clique.
\newblock In {\em Proceedings of the 2013 ACM symposium on Principles of
  distributed computing}, pages 42--50, 2013.

\bibitem[Lub86]{luby86}
M.~Luby.
\newblock A simple parallel algorithm for the maximal independent set problem.
\newblock {\em SIAM Journal on Computing}, 15:1036--1053, 1986.

\bibitem[MS20]{milenkovic2020unified}
Lazar Milenkovi{\'c} and Shay Solomon.
\newblock A unified sparsification approach for matching problems in graphs of
  bounded neighborhood independence.
\newblock In {\em Proceedings of the 32nd ACM Symposium on Parallelism in
  Algorithms and Architectures}, pages 395--406, 2020.

\bibitem[PP22]{PaiP22}
Shreyas Pai and Sriram~V. Pemmaraju.
\newblock Brief announcement: Deterministic massively parallel algorithms for
  ruling sets.
\newblock In Alessia Milani and Philipp Woelfel, editors, {\em {PODC} '22:
  {ACM} Symposium on Principles of Distributed Computing, Salerno, Italy, July
  25 - 29, 2022}, pages 366--368. {ACM}, 2022.

\bibitem[PS95]{panconesi95}
Alessandro Panconesi and Aravind Srinivasan.
\newblock On the complexity of distributed network decomposition.
\newblock {\em Journal of Algorithms}, 20(2):581--592, 1995.

\bibitem[RG20]{rozhovn2019polylogarithmic}
V{\'{a}}clav Rozhon and Mohsen Ghaffari.
\newblock Polylogarithmic-time deterministic network decomposition and
  distributed derandomization.
\newblock pages 350--363, 2020.

\bibitem[SW10]{schneider2010optimal}
Johannes Schneider and Roger Wattenhofer.
\newblock An optimal maximal independent set algorithm for bounded-independence
  graphs.
\newblock {\em Distributed Computing}, 22:349--361, 2010.

\end{thebibliography}

\appendix
\section{Missing Proofs for Degree Reductions}
\label{appendix:deg-reduction}

\subsection{Maximal Independent Sets}

\begin{claim}\label{claim:w1}
    After calling $\OneShotReduceMIS$ with $p=1/\sqrt{\Delta(G)}$, w.h.p, the graph $F$ contains at most $36n$ edges.
\end{claim}
\begin{proof}
    We show that, w.h.p, at most $6np$ vertices were sampled into $V_F$, and that each vertex in the graph has at most $6\Delta(G) p$ neighbors in $F$. In particular, the number of edges in $F$ is at most $36n\Delta(G) p^2=36 n$.
    As each vertex is sampled in $F$ independently, we apply the classic Chernoff bound. Thus, we sample more than $6np$ vertices in $F$ with probability at most  $2^{np}\le \exp(-n/\sqrt{\Delta(G)}) \le \exp(-\sqrt{n})$.
    The probability that a single vertex has more than $6\Delta(G) p$ neighbors in $F$ is $2^{-\Delta(G) p}=2^{-\sqrt{\Delta(G)}}\leq n^{-c}$.
    The two events hold w.h.p.\ for all nodes by union bound.
\end{proof}

\subsection{Degree Reduction for Maximal Matching}

We use the following lemma, which can be seen as the analog of \cref{lemma:deg red restate} for matching.
\begin{lemma}[{\cite[Proposition 3.2]{behnezhad2019fully}}]\label{lemma:deg red edges restate}
    Let $c \ge 1$ be some constant.
    Let $G$ be a graph with $n$ vertices and $m$ edges,
    and $\pi=(e_1,\ldots,e_m)$ be some ordering of the edges.
    Sample each edge independently with probability $p$, and compute the lexicographically first maximal matching with respect to $\pi$ on the sampled edges. Denote the obtained matching by $M$ and let $H=G[V\setminus V(M)]$.
    Then, with probability at least $1-n^{-c+1}$, the residual graph $H$ has maximum degree at most $\frac{c\log (n)}{p}$.
\end{lemma}
Using this lemma, and \cref{lem:max-deg-reduction} we prove \cref{lem:avg-deg-reduction}.

\begin{proof}[Proof of \cref{lem:avg-deg-reduction} for matching using \cref{lemma:deg red edges restate,lem:max-deg-reduction}]
    Let $H$ be the residual graph after sampling each edge in $F$ independently with probability $p=1/d$ and computing the lexicographically first maximal matching on $F$.
    By \cref{lemma:deg red edges restate}, w.h.p, the residual graph $H$ has maximal degree $d(c\log n)$.
    To see that this algorithm can be implemented in \CC, we prove that, w.h.p, the number of sampled edges is $O(n)$.
    The number of edges included in the sampled subgraph is distributed as a binomial variable of parameters $m$ and $p$. It follows that $\Expp{X}=mp \le 2n$.
    Applying the classic Chernoff bound, we find $\Pr{X>6\Expp{X}}\leq 2^{-\Expp{X}}=2^{-n}$.

    To complete the proof, note that the maximum degree in the residual graph $H$ is at most $d(c\log n)$ w.h.p, and therefore by applying one iteration of \cref{lem:max-deg-reduction} we get a residual graph $H_2$ with maximum degree
    \begin{align*}
        \Delta(H_2)\leq \sqrt{\Delta(H)}\cdot \log n \leq \sqrt{d}\cdot \log^2 n\;.
    \end{align*}
    By applying the lemma again on the residual graph $H_2$, we can get rid of the extra $\log n$ factor.
    Note that applying \cref{lem:max-deg-reduction} takes $\BO{1}$ rounds, and therefore the proof follows.
\end{proof}

\section{Opportunistic Routing}
\label{appendix:op-routing}

\ThmLocalSim*

\begin{proof}

    Let $r=C\cdot \brak{\log\Delta + \log^a\log n}$. As $\Delta \le 2^{\sqrt{\log n}/4C}$, for $n$ large enough, we have $r\le 0.5\sqrt{\log n}$.
    In particular, we have $\Delta^{2(r+1)}\leq 2^{\log n/4C + \sqrt{\log n}/2C}\leq \sqrt{n}$. Thus, $\Delta$ and $r$ satisfy \cref{eq:assumption-op-routing} and, by \cref{lemma:op-route}, every vertex learns its $r$-hop neighborhood in $O(1)$ rounds of congested clique. This is sufficient for each vertex to simulate its output in any $r$-round \local algorithm.
\end{proof}

\pagebreak
\subsection{Proof of \cref{lemma:op-route}}
We restate the lemma for convenience.
\ThmOpR*
We start with a description of the algorithm that distributes edges randomly to vertices.
\begin{blackbox}{\OpRoute}
    \begin{description}
        \item[Input]: An integer $r \ge 1$ and a graph $G$ such that \cref{eq:assumption-op-routing} holds.
        \item[Output]: Each vertex $v$ knows all edges of $E^r(v)$, w.h.p. %
    \end{description}

    \begin{enumerate}
        \item Let $v_1, v_2, \ldots, v_n$ be the vertices of $G$. Each vertex $v_j$ samples $n$ incident edges $e_1,\dots,e_n$ uniformly at random with replacement.
              Then, vertex $v_j$ sends $e_i$ to $v_i$. Denote by $E_x$ the set of edges received by $x$ during this round.

        \item Partition vertices into $k=\Delta^{r+1}$ predetermined sets $S_1,\dots S_k$ where $S_i$ for each $i\in[k]$ is defined as
              $$S_i = \set{v_{b_i+1}, v_{b_i+2}, \ldots, v_{b_i + k}} \ , \text{for } b_i=\frac{n(i-1)}{k} \ . $$

        \item
              Consider $S_i$ and a vertex $x\in S_i$.
              Let $G_x=(V,E_x)$ denote the subgraph of $G$ formed by the set of edges received by $x$.
              For every vertex $y\in S_i$, denote the subset of edges in $E_x$ at distance at most $r$ from $y$ in the graph $G_x$ by \[ M_{x,y}\triangleq\set{e\in E_x\mid \dist_{G_x}(y,e)\leq r}. \]
              For each $y\in S_i$, vertex $x$ sends $M_{x,y}$ to $y$.
    \end{enumerate}

\end{blackbox}

\begin{proof}
    The algorithm can be implemented in $O(1)$ rounds of congested clique. During the first step, each vertex sends exactly one message to every other one.
    To bound congestion during the last step, observe that each sets $M_{x,y}$ contain at most $\Delta^{r+1}$ many edges. Hence, each vertex $v\in S_i$ sends/receives at most $|S_i|\Delta^{r+1} \le (n/k)\Delta^{r+1} = n$ messages.

    Fix some vertex $y$ and an edge $e\in E^r(y)$. We bound from below the probability that $y$ obtains the edge $e$.
    Without loss of generality, assume that $y\in S_1$ and fix some $x\in S_1$.
    As $e\in E^r(y)$, there must exists a path $P=(y_1,y_2,\dots, y_{i}, y_{i+1})$ where $y_1=y$ and $e=(y_{i}, y_{i+1})$ for some $i \le r$. %
    For each $j\in[r]$, the probability that $x$ receives the edge $(y_j,y_{j+1})$ from $y_j$ is at least $1/\Delta$.  As edges are sent independently during the first step of the algorithm, the probability that $x$ obtains all edges of $P$ is at least $1/\Delta^{r+1}$.
    Since edges are sent to different vertices independently (because we sample with replacement), the probability that $y$  receives $e$ from some $x\in S_1$ is at least
    \begin{align*}
        1-\brak{1-\frac{1}{\Delta^{r+1}}}^{\abs{S_i}}\geq 1-\exp\brak{-\frac{\abs{S_i}}{\Delta^{r+1}}} \ge 1-\exp\parens*{-\frac{n}{\Delta^{2(r+1)}}} \;.
    \end{align*}
    By a union bound, the probability that there exists some edge which $y$ does not learn about is at most
    \[
        \Delta^{r+1} \exp\parens*{-\frac{n}{\Delta^{2(r+1)}}} \le \Delta^{r+1}\exp\parens*{-(r+1)\log\Delta - c\log n} \le n^{-c} \ ,
    \]
    where the first inequality holds from assumption \cref{eq:assumption-op-routing}. The result follows by another union bound over all possible vertices $y$.
\end{proof}

\section{Concentration Inequalities}

\begin{lemma}[Chernoff Bound {\cite{dubhashi2009concentration}}]\label{thm:6chernoff}
    Let $X_1, \ldots, X_n$ be independent random variables with values in $[0,1]$ and $X:=\sum_i X_i$. If $t\geq 6\Exp{X}$, then
    $\Pr{X\geq t}\leq 2^{-t}$.
\end{lemma}

\begin{definition}[Bounded Differences Property]
    A function $f: \mathcal{X} \rightarrow \mathbb{R}$ for $\mathcal{X}=\mathcal{X}_1 \times \mathcal{X}_2 \cdots \times \mathcal{X}_n$ is said to satisfy the bounded differences property with bounds $c_1, c_2, \ldots, c_n \in \mathbb{R}^{+}$if for all $\bar{x}=$ $\left(x_1, x_2, \ldots, x_n\right) \in \mathcal{X}$ and all integers $k \in[1, n]$ we have
    \begin{align*}
        \sup _{x_k^{\prime} \in X_k}\left|f(\bar{x})-f\left(x_1, x_2, \ldots, x_{i-1}, x_k^{\prime}, \ldots, x_n\right)\right| \leq c_k
    \end{align*}
\end{definition}

\begin{lemma}[{Bounded Differences Inequality \cite[Corollary 5.2]{dubhashi2009concentration}}]
    \label{lemma:bounded difference}
    Let $f: \mathcal{X} \rightarrow \mathbb{R}$ satisfy the bounded differences property with bounds $c_1, c_2, \ldots, c_n$. Consider independent random variables $X_1, X_2, \ldots, X_n$ where $X_k \in \mathcal{X}_k$ for all integers $k \in[1, n]$. Let $\bar{X}=\left(X_1, X_2, \ldots, X_n\right)$ and $\mu=\Expp{f(\bar{X})}$. Then for any $t>0$ we have:
    \begin{align*}
        \Pr{|f(\bar{X})-\mu| \geq t} \leq 2 \exp \left(\frac{-t^2}{\sum_{k=1}^n c_k^2}\right) \;.
    \end{align*}
\end{lemma}

\end{document}